%% file: main.tex
\documentclass[conference]{IEEEtran}
\usepackage[T1]{fontenc}
\usepackage[latin9]{inputenc}
\usepackage[inline]{enumitem}
\usepackage{color}
\usepackage{float}
\usepackage{mathrsfs}
\usepackage{url}
\usepackage{bm}
\usepackage{amsmath}
\usepackage{amsthm}
\usepackage{amssymb}
\usepackage{graphicx}
\usepackage{listings}
\usepackage{multicol}
\usepackage{multirow}
\usepackage{hhline}
\usepackage{colortbl}
\usepackage{stfloats}
\usepackage{pgfplots}
\usepackage{subcaption}
\usepackage{cite}
\usepackage{breakcites}
\usepackage{algpseudocode}

\makeatletter

\floatstyle{ruled}
\newfloat{algorithm}{tbp}{loa}
\providecommand{\algorithmname}{Algorithm}
\floatname{algorithm}{\protect\algorithmname}
\usepackage{tikz}
\usetikzlibrary{calc}

\theoremstyle{plain}

\theoremstyle{plain}

\theoremstyle{plain}

\theoremstyle{plain}
\newtheorem{thm}{\protect\theoremname}

\IEEEoverridecommandlockouts
\author{
\IEEEauthorblockN{Wenjie Liu}
\IEEEauthorblockA{
\textit{Networked Systems Security Group}\\
\textit{KTH Royal Institute of Technology}\\
Stockholm, Sweden}
wenjieli@kth.se
\and
\IEEEauthorblockN{Panos Papadimitratos}
\IEEEauthorblockA{
\textit{Networked Systems Security Group}\\
\textit{KTH Royal Institute of Technology}\\
Stockholm, Sweden}
papadim@kth.se
\thanks{This work was supported in part by the SSF SURPRISE cybersecurity project, the Security Link strategic research center, and the China Scholarship Council.}
}


\usepackage[acronym]{glossaries}
\newcommand{\newac}{\newacronym}
\newcommand{\ac}{\gls}
\newcommand{\Ac}{\Gls}
\newcommand{\acpl}{\glspl}
\newcommand{\Acpl}{\Glspl}

\input{acronym.tex}

\setkeys{Gin}{width=1.0\columnwidth}

\makeatother

\providecommand{\corollaryname}{Corollary}
\providecommand{\lemmaname}{Lemma}
\providecommand{\propositionname}{Proposition}
\providecommand{\theoremname}{Theorem}

\newcommand{\revadd}[1]{\textcolor{.}{#1}}
\definecolor{mycolor1}{rgb}{0.494117647058824,0.184313725490196,0.556862745098039}
\definecolor{mycolor2}{rgb}{0.466666666666667,0.674509803921569,0.188235294117647}
\definecolor{mycolor3}{rgb}{0.301960784313725,0.745098039215686,0.933333333333333}
\definecolor{mycolor4}{rgb}{0.929411764705882,0.694117647058824,0.125490196078431}
\definecolor{mycolor5}{rgb}{0.635294117647059,0.078431372549020,0.184313725490196}
\definecolor{mycolor6}{rgb}{0.8500,0.3250,0.0980}

\begin{document}
\title{Probabilistic detection of GNSS spoofing using opportunistic information}

\maketitle


\setlength\parskip{0pt}

\begin{abstract}
\revadd{\Acpl{gnss}} are integrated into many devices. However, civilian \ac{gnss} signals are usually not cryptographically protected. This makes attacks that forge signals relatively easy. Considering modern devices often have network connections and on-board sensors, the proposed here Probabilistic Detection of \ac{gnss} Spoofing (PDS) scheme is based on such opportunistic information. PDS has at its core two parts. First, a regression problem with motion model constraints, which equalizes the noise of all locations considering the motion model of the device. Second, a Gaussian process, that analyzes statistical properties of location data to construct uncertainty. Then, a likelihood function, that fuses the two parts, as a basis for a \ac{npl}-based detection strategy. Our experimental evaluation shows a performance gain over the state-of-the-art, in terms of attack detection effectiveness. 
\end{abstract}

\begin{IEEEkeywords}
Secure localization, GNSS spoofing detection, opportunistic information
\end{IEEEkeywords}

\glsresetall

\section{Introduction}
\revadd{\Acpl{gnss}} are threatened by a broad gamut of attacks, notably spoofing that allows an adversary to control the position and time information obtained by \ac{gnss} receivers. Attacks were extensively observed recently, e.g., \ac{gps} interference reported by HawkEye 360 \cite{Wer:J22}, attacks on sensor fusion algorithms for autonomous vehicles \cite{SheWonCheChe:C20}, and earlier incidents of luxury yachts being misnavigated \cite{PsiHumSta:J16}. Overall, \ac{gnss} spoofers are getting cheaper and more sophisticated. In response, numerous approaches are proposed to prevent and detect attacks, ranging from the introduction of cryptographic protection, e.g., \cite{WesRotHum:J12,AndCarDevGil:C17}, to signal level mechanisms, e.g., \cite{PapJov:C08,ZhaTuhPap:C15}.

In anticipation of the deployment of such countermeasures and potential residual vulnerabilities, the challenge is how to detect attacks against current \ac{gnss} receivers. A key observation is that \ac{gnss} receivers are typically parts or building blocks of (mobile) computing platforms, which are networked and have integrated \acpl{imu} \cite{SheWonCheChe:C20,KasKhaAbdLee:J22, OliSciIbrDip:J22}. As a result, they have what one can term \emph{opportunistic information} to cross-validate the \ac{gnss} data; notably, information beyond \ac{gnss} that happens to be available, that is, with the help of network interfaces (Wi-Fi, cellular networks, etc.) and on-board sensors (\ac{imu}, wheel speed sensors, etc.). Without considering security, there is already a significant volume of work using opportunistic information for localization. For example, wireless networking infrastructures can provide for alternative positioning methods \cite{Sho:J13,XiaLiuLiChe:J17,LaoMorKimLee:J18,LiuChe:J21}, and the inferred motion based on on-board sensors can be juxtaposed to the \ac{gnss} \ac{pvt} \cite{DixBobKruJac:C20}. Network-based methods often have much larger errors than \ac{gnss} and most on-board sensors (e.g., commercial \ac{imu}) for mobile platforms have a significant cumulative error \cite{XueNiuHonLi:C20}. In contrast, networked-based positioning does not have cumulative errors while on-board inertial sensors do not have large noise fluctuations. 

With opportunistic information readily available, \ac{gnss} spoofing attack detection mechanisms can be designed. In \cite{OliSciIbrDip:C19}, \ac{gnss} devices with cellular connections perform \ac{wcl} using measured \ac{rss} data. The distance between the estimated cellular location and \ac{gnss} enables a binary decision (on \ac{gnss} being attacked). This is extended in \cite{OliSciIbrDip:J22} by adding Wi-Fi and different communication situations, without on-board sensors. They use a simple and direct way to make decisions on spoofing (i.e., Euclidean distance). \Ac{imu}-based detection \cite{CecForLauTom:J21} requires the device to perform maximum likelihood estimation to get position, acceleration, and velocity. Then, combined with orientation into a trajectory vector, a \ac{glrt} can detect spoofing.

Extending the attack detection beyond a binary outcome and basing it upon all available opportunistic information, not only network-based or on-board sensor-based, is straightforward. To the best of our knowledge, it remains unexplored and it is the focus of this work. We consider all existing opportunistic information (networks and on-board sensors) to provide a likelihood of \ac{gnss} under spoofing. We combine positioning based on available terrestrial networks (Wi-Fi, cellular, etc.) and on-board inertial sensors to implement a robust and efficient Probabilistic Detection of \ac{gnss} Spoofing (PDS) scheme. The key idea is to consider the different kinds of noisy measurements and the continuity of the motion model of the \ac{gnss} receiver, and then build a probability space of positions and aggregate it in a weighted manner into one \emph{likelihood function}. It is important to note that we assume that opportunistic information is noisy but not subject to attacks, consistent with the literature \cite{KasKhaAbdLee:J22, OliSciIbrDip:J22}); we discuss future work towards extending the adversary model. 

The likelihood function construction has two parts. First, we derive a closed-form solution for the relation between motion and position information and propose a motion-constrained regression problem \cite{Fan:b96}; this balances short-term estimation through the receiver movement and long-term estimation through network signals, so the location data is smoothed based on motion model constraints. Second, a Gaussian process regression \cite{SchSpeKra:J18} models the uncertainty of the smoothed locations. Finally, we combine both parts (i.e., motion and statistics) into one test statistic by using a weighted sum, similar to the \ac{glrt} \cite{RotCheLoWal:J21}, and make decisions based on it. The detection we propose here can detect any type of spoofing, e.g., gradual deviation, replay attack, etc., because it operates based on the resultant \ac{gnss}-provided (possibly attacked) position. 

Our main contributions are: a multi-sided opportunistic information-based \ac{gnss} spoofing detection method using on-board inertial sensors and locations from terrestrial network infrastructures. We fuse the heterogeneous data using model-based (motion-constrained local polynomial regression) and model-free (Gaussian process) methods. We handle the one-time and cumulative observational errors through a constrained optimization problem that considers motion and network-based locations jointly. The likelihood function used in decision-making incorporates weights and integrates information from both temporal and categorical perspectives. We conduct simulation-based experiments building on three datasets to evaluate the performance advantage in terms of spoofing detection true positive rate, attack detection delay (the time between the attack launch and detection), and the accuracy of the positioning even when under a \ac{gnss} attack. 

The rest of the paper is organized as follows. \revadd{Sec.~\ref{relwor} mentions work related to positioning and \ac{gnss} spoofing.} Sec.~\ref{sysmod} and \ref{prosta} introduce our system model, attack model, assumptions, and problem formulation. Sec.~\ref{schout} is an overview of the probabilistic \ac{gnss} spoofing detection using opportunistic information. Sec.~\ref{winrol}, Sec.~\ref{concon}, and Sec.~\ref{croseq} develop the algorithms and introduce the theoretical results for three components, window rolling, confidence intervals, and decision-making. \revadd{Numerical results and performance comparison with related work are discussed and presented in Sec.~\ref{numres}.} Finally, the conclusion is drawn in Sec.~\ref{conclu}.


\section{Background and Related Work}
\label{relwor} 
\subsection{Terrestrial Positioning}
Terestrial infrastructures, e.g., Wi-Fi, cellular, and eLoran \cite{Gow:J20}, can play a role as \ac{gnss} backup. Terrestrial positioning usually uses matrix completion \cite{XiaLiuLiChe:J17}, fingerprinting, and range-based methods \cite{LaoMorKimLee:J18}. Matrix-completion-based localization does not depend on range measurements, which may be incomplete and corrupted. Noise and multi-norm regularization is used for matrix completion, such that the matrix is unimodal and the position of its peak is the location estimation. Fingerprint-based methods collect a database of fingerprinting in advance, including signal strengths, magnetic field, channel state information, and so on. Then, deterministic or probabilistic fingerprinting matching algorithms are used for localization. Range-based methods use signal strength, propagation time, or angle of arrival to derive pseudoranges, for multilateration. 
\subsection{GNSS Spoofing Attacks}
A \ac{gnss} spoofer typically generates a counterfeit \ac{gnss} signal with the right power and format, according to the specifications. Before carrying out a spoofing attack, the attacker can use jamming to intentionally interfere with \ac{gnss} signals, forcing the victim \ac{gnss} receiver to loose the signal lock \cite{KasKhaAbdLee:J22} or, at the expense of complexity, it can gradually lift the adversarial signal and make eventually the victim track it \cite{SheWonCheChe:C20}. The simplest way of generating adversarial signals is meaconing, the retransmission of legitimate \ac{gnss} signals with a time delay. A more advanced variation of meaconing, selective delay, can rebroadcast individual satellite signals \cite{PapJov:C08}, then modify the position solution according to the attack scenario. Replay/relay attacks can be mounted with low-complexity setups \cite{LenSpaPap:C22}. 
Distance-decreasing (DD) attacks \cite{ZhaPap:C19b} can give more options to the adversary, using Early Detection and Late Commit to relay the signal, resulting in the forged signal seemingly arriving earlier than the actual signal would have arrived. The challenge lies in that replay/relay attacks can be effective even if there is cryptographic protection \cite{PapJov:C08,ZhaPap:C19b,LenSpaPap:C22}.
\subsection{GNSS Spoofing Detection}
Traditional methods of spoofing detection dig into the characteristics of signals doing anomaly detection with time series, such as Doppler effect, \ac{rss}, \ac{snr}, and \ac{aoa} \cite{PapJov:C08,BroJafDehNie:C12,ZhaPap:C19}. At the same time, considering \ac{gnss} receivers as part of mobile computing platforms, integrated network interfaces and on-board sensors \cite{KasKhaAbdLee:J22} can provide for information to detect attacks. Terrestrial network infrastructures, already considered as parts of a backup positioning solution, can provide location information that can be used to assess the \ac{gnss}-provided position \cite{OliSciIbrDip:J22}. 
In \cite{KhaRosLanCha:C14}, the authors proposed schemes using \ac{imu} to cross-check the \ac{gnss} location. They use a Kalman filter to fuse \ac{gnss} states and \ac{imu} measurements, while \ac{raim} performs the spoofing detection. In \cite{MicForCenTom:J22}, a fused location is derived from \ac{gnss} information with \ac{imu}, and compared the relative distance information from \ac{rss} data to detect spoofing. Their alternative location can further help the robust navigation under spoofing. 

\section{System Model and Adversary}
\label{sysmod}
\subsection{System Model}
\revadd{We consider a mobile \ac{gnss} platform (e.g., smartphone, car, and done) equipped with common modules providing opportunistic information, including network interfaces (e.g., Wi-Fi and cellular networks) and on-board sensors (e.g., \ac{imu} and speed sensors).} Some network modules may be unavailable, interfered with by benign signals or perform with an unacceptably high network latency. \ac{gnss}, Wi-Fi, and cellular networks provide updated locations based on existing positioning algorithms \cite{MagLunPat:J15,ShaKas:J21}. \Acpl{imu} provide multi-axis acceleration measurements. In some situations, e.g., in vehicles, we can get speed from the wheel speed sensors. The opportunistic information is not specially designed for spoofing detection, and it is updated as a discrete time series at a specific frequency. When the mobile platform in Fig.~\ref{fig:locations} moves and navigates on a path in a benign environment, the \ac{gnss}-provided location should be consistent with the opportunistic information. Similarly, under a \ac{gnss} attack, the \ac{gnss}-provided location would typically deviate from the actual location and be inconsistent with the opportunistic information. 

\textbf{Notation.} $\mathbf{p}_{\text{c}}(t) \in \mathbb{R}^2$ is the receiver actual location at time $t$, which needs to be determined based on positioning. \Ac{gnss} and network-based positions are $\mathbf{p}_m(t)$. $\mathbf{p}_0(t)$ is the \ac{gnss} position at time $t$. $\mathbf{p}_m(t),m=1,2,...,M$ are locations from networks, where $M$ is the number of network interfaces; $t=1,2,...,N$ in second and $N$ is the last time index. $\mathbf{p}_m(t)$ are also with an unavailability probability $U_m$. Regarding on-board sensors: they provide speed, $\mathbf{v}(t)$, acceleration, $\mathbf{a}(t)$, and angular rate, $\boldsymbol{\omega}(t)$. We assume that positioning errors based on benign settings for \ac{gnss}, Wi-Fi, and cellular networks are Gaussian random variables. 
\begin{figure}
\begin{centering}
\includegraphics[width=\columnwidth]{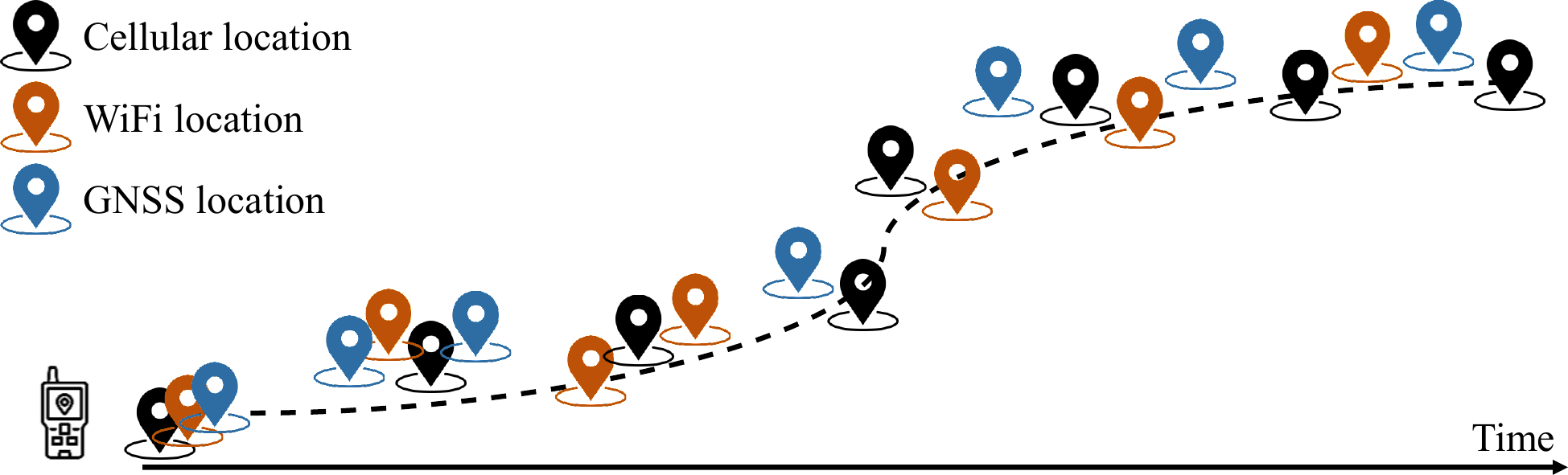}
\par\end{centering}
\caption{A two-dimensional example of location information from the existing infrastructures.}
\label{fig:locations}
\end{figure}
\subsection{Adversary}
Fig.~\ref{fig:advers} provides a high-level illustration of the adversary model: spoofed or replayed/relayed \ac{gnss} signals force the mobile platform wrongly estimate its position. We are agnostic to the attack specifics but we do not constrain the attacker. We assume it knows the victim location with almost practically no observational error and it can use a \acpl{sdr} with state-of-the-art \ac{gnss} spoofing capabilities to falsify the position. 
The attacker can delicately design the trajectory of the victim (spoofed \ac{gnss}) location, concerning the real path of the mobile platform. Some trajectory strategies, such as gradual deviation \cite{SheWonCheChe:C20} and path drift \cite{NarRanNou:C19}, can be imperceptible for a period after the onset of the attack. 
The attack can be either in \emph{cold start} or force the victim to loose lock on legitimate signals and lock on adversarial signals. 
\begin{figure}
\begin{centering}
\includegraphics[width=\columnwidth]{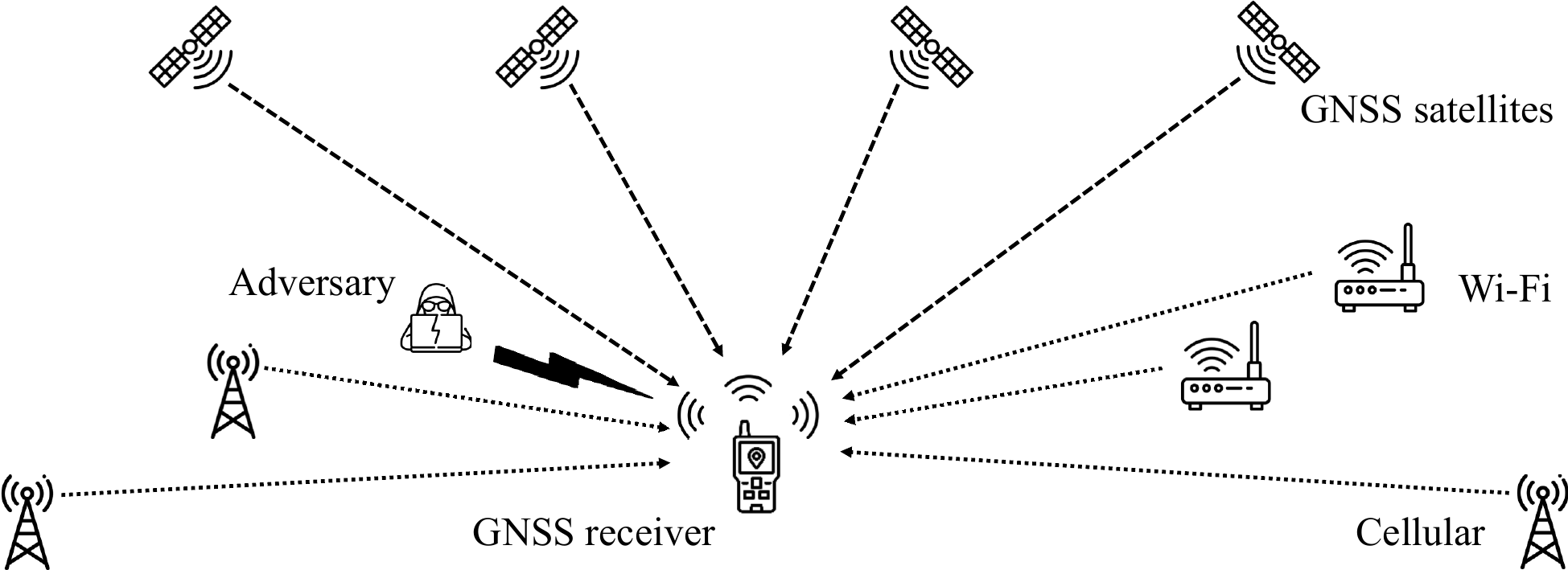}
\par\end{centering}
\caption{System and adversary model illustration.}
\label{fig:advers}
\end{figure}

We assume the attacker operates only in the \ac{gnss} domain but does not attack other networks and thus does not affect the resultant positioning. It is assumed in the context of this work that the adversary can only jam wireless networks. This can prolong the periods of unavailability for alternative positioning. 
In addition, we assume the attacker does not physically control the victim, thus the procedure of deriving location information from different network interfaces and on-board sensors can not be manipulated. 

\section{Problem Statement}
\label{prosta}
We aim to test if the \ac{gnss} provided position is consistent with opportunistic information, and accordingly make a decision on whether the current \ac{gnss} position is the result of an attack. Based on the likelihood of \ac{gnss} being under attack, we want maximize the accuracy of \ac{gnss} spoofing detection. We also aim for a system operating even if some type, not all types, of opportunistic information is unavailable. We focus on the design of the detection algorithm. 

For \ac{gnss} spoofing detection at a time $t$, we have data $\{\mathbf{p}_m(i),{\mathbf{v}}(i),{\mathbf{a}}(i),\boldsymbol{\omega}(i) \}$ for all $0<i<t$ and $m \in \{0,1,...,M\}$, based on which it is decided if $\mathbf{p}_0(t)$ is the result of a \ac{gnss} attack. The hypotheses are: 
\begin{itemize}
    \item $H_0$: \ac{gnss} is not under attack.
    \item $H_1$: \ac{gnss} is under attack.
\end{itemize}
The decision at time $t$ is $\hat{H}(t)\in \{H_0,H_1\}$. The true positive for $t$ can be written as $\mathbb{I}\{\hat{H}(t)=H_1|H_1\}=1$ and the false positive as $\mathbb{I}\{\hat{H}(t)=H_1|H_0\}=1$, where the indicator function $\mathbb{I}\{A|B\}$ takes value 1 if A holds on condition of B. We denote the total number of positives as $N_\text{P}$, the number of true positives as $N_\text{TP}$, and the number of false positives as $N_\text{FP}$. The true positive probability for the period $0<t \le N$ is: 
\begin{equation}
    P_\text{TP} (\hat{H}(t)) = P(\hat{H}(t)=H_1|H_1)=\frac{N_\text{TP}}{N_\text{P}}.
\end{equation}
The Type I error (false positive) probability is: 
\begin{equation}
    P_\text{FP} (\hat{H}(t)) = P(\hat{H}(t)=H_1|H_0)=\frac{N_\text{FP}}{N-N_\text{P}}.
\end{equation}

We define the detection time delay $\Delta T$ as the difference between the time of raising alarm and attack being launched: 
\begin{eqnarray*}
\Delta T&=&\min\left\{ t \Bigm|\mathbb{I}\{\hat{H}(t)=H_{1}|H_{1}\}=1\right\} \\&&-\min\left\{ t \Bigm|\mathbb{I}\{\hat{H}(t)=H_{0}|H_{1}\}=1\right\} 
\end{eqnarray*}

The goal of this work is to: (a) maximize the true positive probability $P_\text{TP}$, given a maximum allowable false positive probability, $P_{\text{FP}_{\max}}$:
$$
\begin{matrix}
	\max&P_\text{TP} (\hat{H}(t))\\
	\text{s.t}.&P_\text{FP} (\hat{H}(t)) \le P_{\text{FP}_{\max}}\\
\end{matrix}.
$$
and (b) minimize the detection time delay $\Delta T$:
$$
\begin{matrix}
	\min&\Delta T (\hat{H}(t))\\
	\text{s.t}.&P_\text{FP} (\hat{H}(t)) \le P_{\text{FP}_{\max}}\\
\end{matrix}.
$$
and (c) provide a likelihood in $[0,1]$ of \ac{gnss} being under attack along with (d) an alternative position.

\section{Proposed Scheme}
We start with the outline of the proposed PDS scheme, then introduce three system components and three application cases, i.e., use of: (i) network interfaces only, (ii) on-board sensors only, and (iii) network interfaces and on-board sensors.
\subsection{Scheme Outline}
\label{schout}
\begin{figure}
\begin{centering}
\includegraphics[width=\columnwidth]{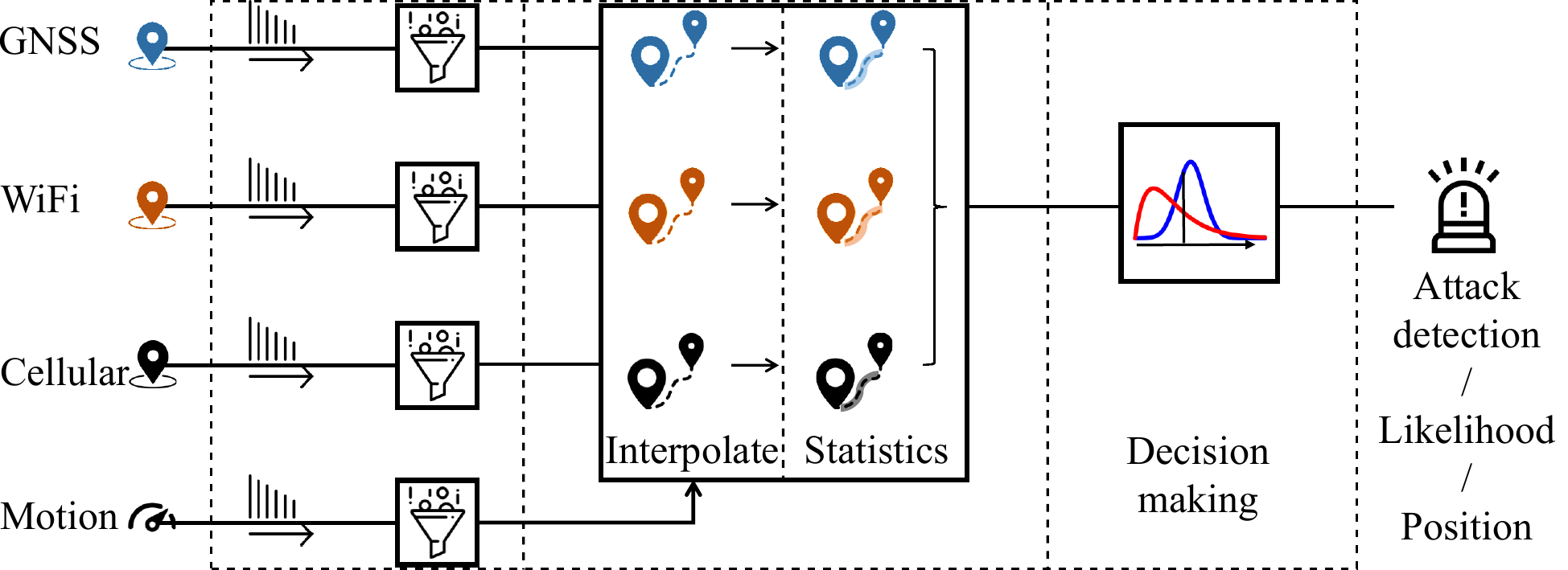}
\par\end{centering}
\caption{System overview of PDS.}
\label{fig:scheme}
\end{figure}
PDS aims at detecting \ac{gnss} spoofing through multiple information sources of opportunity that can provide locations, speed, and acceleration. We provide a high-level description of how the \ac{gnss} spoofing detection system works, illustrated in Fig.~\ref{fig:scheme}. The system collects input data from \ac{gnss} and opportunistic information sources. The filters perform window rolling to screen $\mathbf{p}_m(i)$ and, with the help of $\{{\mathbf{v}}(i),{\mathbf{a}}(i),\boldsymbol{\omega}(i) \}$, we interpolate the fixed-length series of positions to a continuous function. Then, \revadd{we model the confidence intervals of the three series of positions.} Through the fused confidence intervals, the decision module constructs the likelihood and allows us to decide if the current \ac{gnss} position is the result of an attack and if so raise an alarm. 

PDS mainly relies on three components: window rolling, confidence interval construction, and a likelihood function that is fused from time and category perspectives. These three components are essential in processing the location series, establishing the confidence intervals for the locations, and making decisions based on the likelihood function. 

The first component collects positions from \ac{gnss}, Wi-Fi, and cellular in real-time ($M$ alternative sources, here we deal with $M=2$), as well as speed and acceleration from the \ac{imu}. The data points are sorted and arranged into time series by timestamp. Then, the filters perform window rolling to extract fixed-size sequences at each time instead of using the whole data. The window size is chosen by cross-validation, achieving the best performance on the previous empirical data. 

The second component estimates the confidence intervals based on the series from the motion and statistical models. 
\begin{itemize}
\item The motion model leverages short-term and long-term characteristics for estimation. $\mathbf{v}(t)$, $\mathbf{a}(t)$, $\boldsymbol{\omega}(t)$, are always updated at a high frequency and with almost no latency compared to obtaining locations from network infrastructures. However, they are short-term accurate and \acpl{imu} cannot directly provide locations. Moreover, the integral of these raw data that can be used to indirectly calculate the current location is not accurate in the long term. The mobile platform can be localized based on the terrestrial networks, in principle at a lower accuracy and frequency compared to \ac{gnss}. On the one hand, their one-time positioning error is relatively high, without, an accumulated error as that for \ac{imu}. Such positioning can be performed periodically, though not highly accurately, and we propose to combine it with short-term \ac{imu} accuracy. 
\item The statistical model is concerned with confidence intervals in the form of a probability distribution representing the uncertainty of \ac{gnss} spoofing with respect to location, with the mean value being the output of the motion model estimations. The Gaussian process model uses a predefined covariance function and estimates the variance of the probability distribution. 
\end{itemize}
The motion model uses a local polynomial regression algorithm to fit positions at time $t$, while it satisfies the constraint of the movement. The fitting algorithm assigns weights to different data points (\revadd{the more recent the data, the higher the weight}). Then, it fits the positions and minimizes the fitting error. At the same time, the movement constraint ensures the fitted result satisfies the physical feasibility of speed, acceleration, and attitude. Based on the statistical model, we analyze locations, where we make the assumption that the location series follows a Gaussian process with a mean derived from local polynomial regression. To calculate the variance, we utilize Gaussian process regression and employ the differences between the fitted results and input data points. This helps us to obtain a more precise measure of the variance associated with the location data. 

The third component constructs a likelihood function using the confidence intervals in the form of probability distribution with mean and variance, and decision-making relies on calculating the likelihood ratio from the time domain of the intervals and across different sources. In the time domain, we combine the confidence intervals by computing a weighted sum over time. The weights used in the computation of the sum are normalized to ensure that they add up to one. The resulting distribution is also Gaussian, with its mean and variance expressed as a linear combination of the original means and variances of the individual distributions. We also combine the $M$ alternate position sources and the \ac{gnss}-obtained one, by multiplying $M+1$ distributions at time $t$. Then, we construct a likelihood function based on the combined distribution and use the \ac{npl} to get a threshold for maximizing the true positive when fixing the false positive, $P_{\text{FP}_{\max}}$, of the \ac{gnss} spoofing detection. 
\subsection{Rolling for Screening and Detection}
\label{winrol}
We use a simultaneous screening and detection strategy in the algorithm. Window rolling allows one to choose a fixed-length sequence, $S$, instead of using the entire time series for every estimation at each time slot $t$, for efficiency. In addition, outdated data is not helpful. Hence, we use a window rolling strategy with a window size to filter it out. 
\subsubsection{Data Collection}
We collect $\mathbf{p}_m(t) \in \mathbb{R}^2, m=0,1,...,M$, in the format of World Geodetic System 1984 (WGS84)\footnote{WGS 84, the coordinate reference system used by \ac{gps}, is an oblate spheroid surface centered at the center of the earth with the parameters equatorial ($a$), poles ($b$), and inverse flattening ($1/f$).} data. $\mathbf{v}(t),\mathbf{a}(t) \in \mathbb{R}^3$ use a right-forward-up coordinate system\footnote{Right-forward-up coordinate system uses right, forward, and up of the mobile platform as Cartesian coordinates.}. $\boldsymbol{\omega}(t) \in \mathbb{R}^3$ consists of roll $\phi$, pitch $\theta$, and yaw $\psi$, which are angles between local coordinates and WGS84, calculated by a gyroscope and a magnetometer. 
\subsubsection{Window Size}
This $w$ is a model parameter, the length of $S$, the most recent benign sequence: $S=\{ \mathbf{p}_m(i),{\mathbf{v}}(i),{\mathbf{a}}(i),\boldsymbol{\omega}(t) \}$ for $t-w<i<t$. There are various methods to determine the optimal rolling window size, including using cross-validation to minimize the \ac{mse}, and they still need to be explored in our future work. Once a suitable window size $w$ is chosen, we can proceed to process the sequence $S$.
\subsubsection{Processing $S$}
During the initialization, benign data series $S$ is the initial input. The dimension of $S$ is $w \times (M+5)$, including timestamp, $M+1$ location sources, speed, acceleration, and angular rate. Then, for each time $t$, we check the current \ac{gnss} position and decide whether it is under attack. If the likelihood of being under attack is higher than a threshold, the system raises an alarm and updates $S$ by using $\{  \mathbf{p}_m(t),{\mathbf{v}}(t),{\mathbf{a}}(t),\boldsymbol{\omega}(t) \}, m \in \{1,2,...,M\}$ without \ac{gnss}. If not under attack, we update $S$ by using information from all sources, $\{  \mathbf{p}_m(t),{\mathbf{v}}(t),{\mathbf{a}}(t),\boldsymbol{\omega}(t) \}, m \in \{0,1,...,M\}$. The overall process is shown as Algorithm \ref{alg:winrol}.
\begin{algorithm}
\hspace*{\algorithmicindent} \textbf{Input} $t,\mathbf{p}_0(t),\mathbf{p}_1(t),\mathbf{p}_2(t),{\mathbf{v}}(t),{\mathbf{a}}(t),\boldsymbol{\omega}(t)$\\
\hspace*{\algorithmicindent} \textbf{Parameter} $w$\\
\hspace*{\algorithmicindent} \textbf{Output} \textit{IsAttack}
\begin{algorithmic}[1]
\State \textbf{initialize} $S$ as a matrix of input data
\While{$t \leq N$}
\State $t \gets t+1$
\State \textbf{ensure} $length(S) \leq w$
\State \textit{CI}=\textit{ConstructCI}($\mathbf{p}_0(t), S$) \Comment{Algorithm \ref{alg:confiden}}
\State \textit{IsAttack}=\textit{MakeDecision}(\textit{CI}) \Comment{Algorithm \ref{alg:decision}}
\If{\textit{IsAttack}} 
    \State $S \gets S+\{ \mathbf{p}_m(i),{\mathbf{v}}(i),{\mathbf{a}}(i),\boldsymbol{\omega}(t) \},m=1,2,...$
\Else
    \State $S \gets S+\{ \mathbf{p}_m(i),{\mathbf{v}}(i),{\mathbf{a}}(i),\boldsymbol{\omega}(t) \},m=0,1,2,...$
\EndIf 
\EndWhile
\end{algorithmic}
\caption{Window rolling for screening and detection \label{alg:winrol}}
\end{algorithm}

\subsection{Constructing the Confidence Interval}
\label{concon}
First, locations are smoothed by a motion-constrained regression problem \cite{Fan:b96} as illustrated in Fig.~\ref{fig:interpolation}. Second, a Gaussian process \cite{SchSpeKra:J18} models the confidence intervals as Fig.~\ref{fig:gp_part}. This is summarized in Algorithm \ref{alg:confiden}. 
\begin{algorithm}
\hspace*{\algorithmicindent} \textbf{Input} $S$\\
\hspace*{\algorithmicindent} \textbf{Output} \textit{CI}
\begin{algorithmic}[1]
\State $i \gets 0$
\While{$i \leq w$}
    \State $i \gets i+1$
    \State $\hat{\mathbf{p}}_m(t-w+i) \gets$ Eq.~\eqref{eq:proall} \Comment{Motion part}
    \State $\mathbf{x}_{t-w+i} \gets \hat{\mathbf{p}}_m(t-w+i)-\mathbf{p}_m(t-w+i)$
    \State \Comment{Statistics part}
\EndWhile
\State \textit{CI} $\gets$ Eq.~\eqref{eq:conint}
\end{algorithmic}
\caption{Construct the spoofing confidence interval \label{alg:confiden}}
\end{algorithm}

\subsubsection{Motion Model}
We use local polynomial regression to interpolate and estimate the position for continuous time, based on discrete $\mathbf{p}_m(t)$ location points. The local polynomial regression is an attractive method of non-parametric regression and fits Taylor expansion of an unknown function at a point by a weighted least squares regression. So, for time $t$, with polynomial order $n$, the estimator $\hat{\mathbf{p}}_m(t)$ for $\mathbf{p}_m(t)$ is 
\[
\hat{\mathbf{p}}_m(t) = \mathbf{W}\mathbf{t}
\]
where $\mathbf{W} \in \mathbb{R} ^{2 \times (n+1)}$ is a matrix of polynomial coefficients, $\mathbf{t}$ is a $(n+1)$ dimensional vector and $[\mathbf{t}]_i=t^{i-1}$. 

\textbf{Case 1.} With network-based positions only, our optimization problem for estimating $\hat{\mathbf{p}}_m$  for any $m>0$ at time $t'$ is
\[\begin{array}{*{20}{c}}
  {\mathop {\min }\limits_{\mathbf{W}} }&{\sum\limits_{t=t'-w}^{t'-1} [\mathbf{W} \mathbf{t}-\mathbf{p}_m(t)]^\top K_\text{loc}(t-t')[\mathbf{W} \mathbf{t}-\mathbf{p}_m(t)]} 
\end{array}\]
where $K_\text{loc}$ is a kernel function assigning weights. 

\textbf{Case 2.} For on-board sensor data, the coordinate systems need to be unified. $\mathbf{R}$ is the rotation matrix that transforms the local right-forward-up coordinate system to WGS84 coordinates:
\begin{align*}
\mathbf{R}(t)	&=	\mathbf{R}_{\psi}(t)\mathbf{R}_{\theta}(t)\mathbf{R}_{\phi}(t) \\
	&=	\left[\begin{array}{ccc}
\cos\psi(t) & -\sin\psi(t) & 0\\
\sin\psi(t) & \cos\psi(t) & 0
\end{array}\right] \\
& \qquad \times \left[\begin{array}{ccc}
\cos\theta(t) & 0 & \sin\theta(t)\\
0 & 1 & 0\\
-\sin\theta(t) & 0 & \cos\theta(t)
\end{array}\right] \\
& \qquad \qquad \times \left[\begin{array}{ccc}
1 & 0 & 0\\
0 & \cos\phi(t) & -\sin\phi(t)\\
0 & \sin\phi(t) & \cos\phi(t)
\end{array}\right].
\end{align*}
We denote the state of the mobile platform as $\big(\mathbf{p}_m(t),\mathbf{v}(t),\mathbf{a}(t) \big)$, so
\begin{align*}
\tilde{\mathbf{p}}_m(t)&=\mathbf{p}_m(t-\Delta t)+\mathbf{R}(t-\Delta t)\mathbf{v}(t-\Delta t)\Delta t\\
&\qquad+\frac{1}{2}\mathbf{R}(t-\Delta t)\mathbf{a}(t-\Delta t)(\Delta t)^2\\
\tilde{\mathbf{v}}(t)&=\mathbf{R}(t-\Delta t)\mathbf{v}(t-\Delta t)\Delta t\\
&\qquad+\mathbf{R}(t-\Delta t)\mathbf{a}(t-\Delta t)\Delta t
\end{align*}
Then the state-transition model is
\begin{equation}
    \mathbf {F}(t) ={\begin{bmatrix}\mathbf{1}&\mathbf{R}(t) \Delta t\\\mathbf{0}&\mathbf{1}\end{bmatrix}}.
\end{equation}
The control-input model is 
\begin{equation}
    \mathbf {B}(t) ={\begin{bmatrix}{\mathbf{R}(t)\Delta t^2}/{2}\\\mathbf{R}(t) \Delta t\end{bmatrix}}.
\end{equation}
We conclude that $\big(\tilde{\mathbf{p}}_m(t),\tilde{\mathbf{v}}(t)\big)=\mathbf {F}(t-\Delta t) \cdot \big(\mathbf{p}_m(t-\Delta t),\mathbf{v}(t-\Delta t)\big)+\mathbf{B}(t-\Delta t)\mathbf{a}(t-\Delta t)+\mathbf {n} $, where $\mathbf {n}$ models unknown effects. 

Then, our optimization problem for estimating a continuous $\hat{\mathbf{p}}_0$ at time $t'$ is
\begin{equation}
    \begin{array}{*{20}{c}}
    {\mathop {\min }\limits_{\mathbf{W}} }&{\sum\limits_{t=t'-w}^{t'-1} [\mathbf{W} \mathbf{t}-\mathbf{p}_0(t)]^\top K_\text{loc}(t-t')[\mathbf{W} \mathbf{t}-\mathbf{p}_0(t)]} \\ 
    {{\text{s}}{\text{.t}}{\text{.}}}&{|\mathbf{W} \mathbf{t'} - \tilde{\mathbf{p}}_0(t' )| \le \Delta t \cdot \boldsymbol{\epsilon}}
    \end{array}
\label{eq:promot}
\end{equation}
where $\Delta t \cdot \boldsymbol{\epsilon} \in \mathbb{R}^2$ in the constraint is a noise tolerance term based on the length of the time slot, in order to satisfy the requirement that both long-term and short-term (on-board sensors) anti-spoofing should be considered in the design. In addition, if the on-board sensor does not provide speed, $\mathbf{v}(t)$ is replaced by $\intop_{0}^{t}\mathbf{a}(t)\textrm{d}t$ or other approximations. Similarly, if the on-board sensor does not provide acceleration, $\mathbf{a}(t)$ is set to zero, which means the movement is seen as uniform motion in a short period. 

\textbf{Case 3.} For the combination of networks-based positioning and on-borad sensors, the optimization problem is 

\begin{equation}
    \begin{array}{*{20}{c}}
    \mathcal{P}: \\
    {\mathop {\min }\limits_{\mathbf{W}} }&{\sum\limits_{t=t'-w}^{t'-1} [\mathbf{W} \mathbf{t}-\mathbf{p}_m(t)]^\top K_\text{loc}(t-t')[\mathbf{W} \mathbf{t}-\mathbf{p}_m(t)]} \\ 
    {{\text{s}}{\text{.t}}{\text{.}}}&{|\mathbf{W} \mathbf{t'} - \tilde{\mathbf{p}}_m(t' )| \le \Delta t \cdot \boldsymbol{\epsilon}}
    \end{array}
\label{eq:proall}
\end{equation}
that $\forall m$. After the solution of $\mathcal{P}$ is presented, we can calculate the deterministic estimation of position at time $t'$. 
\begin{thm}
The estimator $\hat{\mathbf{p}}_m(t)$ can estimate $\mathbf{p}_m(t)$ in polynomial time, i.e., the problem $\mathcal{P}$ is a polynomial time problem.
\end{thm}
\begin{proof}
We take the second derivative of the objective function of $\mathcal{P}$ with respect to $\mathbf{W}$: 
\[
2 \sum\limits_{t=t'-w}^{t'-1}  K_\text{loc}(t-t')\cdot (\mathbf{t'}\cdot \mathbf{t'}^\top )^\top \otimes \mathbb{I}
\]
which is a positive definite matrix, as $K_\text{loc}(t-t')>0$ always holds. Hence, the objective function is convex. The constraints in \eqref{eq:promot} and \eqref{eq:proall} are equivalent to
$$
\left\{ \begin{array}{l}
	\mathbf{W} \mathbf{t'} - \tilde{\mathbf{p}}_m(t' ) \le \Delta t\cdot \boldsymbol{\epsilon }\\
	\mathbf{W} \mathbf{t'} - \tilde{\mathbf{p}}_m(t' ) \ge -\Delta t\cdot \boldsymbol{\epsilon }\\
\end{array} \right. ,\forall t
$$
which are affine functions. Hence, $\mathcal{P}$ is a convex optimization problem. \revadd{It is solvable using Lagrange multipliers, thus} the estimator $\hat{\mathbf{p}}_m(t)$ can estimate $\mathbf{p}_m(t)$ in polynomial time. 
\end{proof}

With an optimization problem $\mathcal{P}$ (Eq.~\eqref{eq:proall}) solvable in polynomial time, the prediction can be real-time in practice. By recursively updating the state estimate, the location sequence interpolation result $\hat{\mathbf{p}}_m(t)$ should be similar to the solid lines shown in Fig.~\ref{fig:interpolation} that show the interpolated locations, which is the estimated trace of the mobile platform. 
\begin{figure}
\begin{centering}
\includegraphics[trim={0 1cm 0 0},clip,width=\columnwidth]{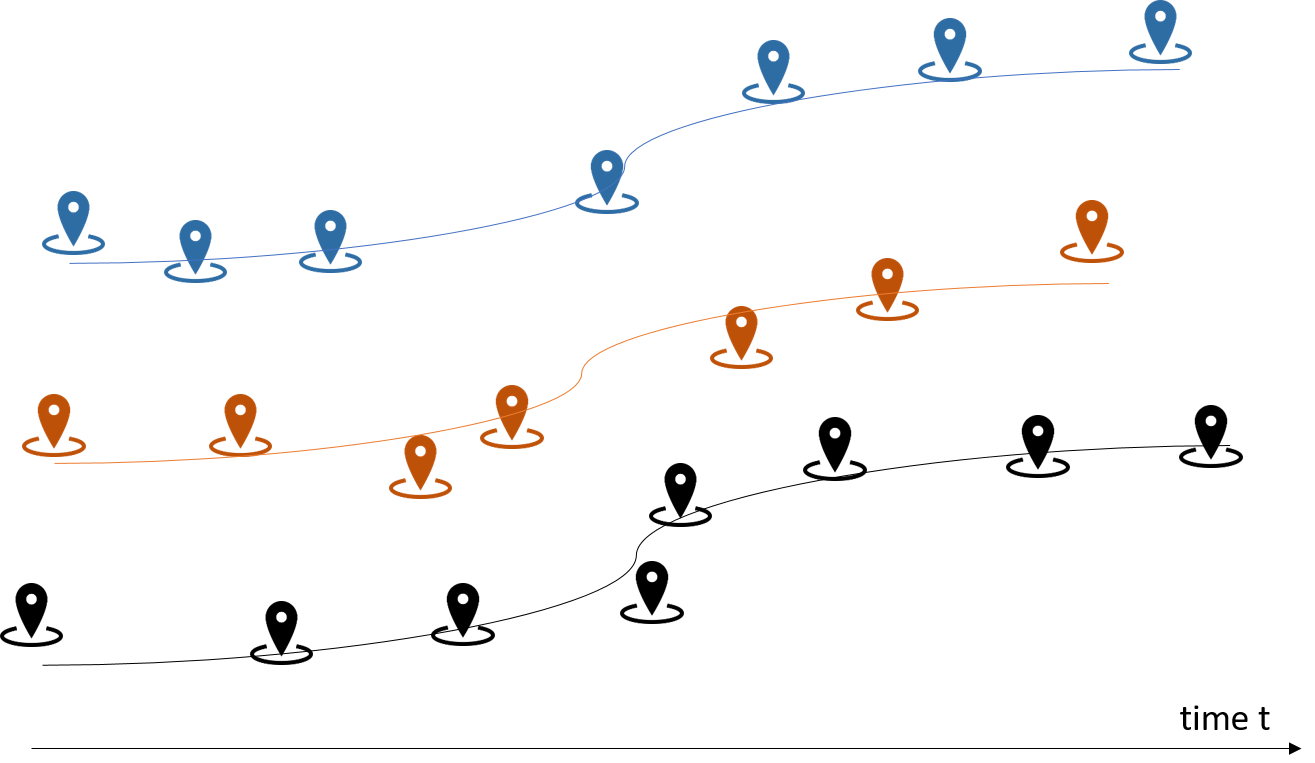}
\par\end{centering}
\caption{Estimated trace from local polynomial regression (tiled view).}
\label{fig:interpolation}
\end{figure}

\subsubsection{Statistical Model}
We use the Gaussian process on the trace from the previous step to model the residual part of estimated positions. Suppose the residual part at time $t$ is:
\begin{equation}
    \mathbf{x}_m(t)=\hat{\mathbf{p}}_m(t)-\mathbf{p}_m(t).
\end{equation}
Then, $\left\{\mathbf{x}_m(i);i\in (0,t)\right\}$ are zero-mean Gaussian random variables. The covariance function $K(\mathbf{x},\mathbf{x}')=\frac{1}{2}\mathbb{E}[(\mathbf{x}-\mathbf{x}')^2]$ will be chosen to describe its interrelation \cite{SchSpeKra:J18}. Then, a linear unbiased estimator can predict the current $\mathbf{x}_m(t)$: 
\begin{equation}
    \hat{\mathbf{x}}_m(t)=\sum_{i=t-w}^{t-1} \lambda_i \mathbf{x}_m(i)
\end{equation}
where $\sum_{i=t-w}^{t-1} \lambda_i=1$. Gaussian process regression determines $\lambda_i$ that minimizes the variance of the estimation error, 
\begin{equation}
    \begin{array}{*{20}{c}}
  {\mathop {\min }\limits_{\mathbf{\boldsymbol{\lambda}}} }&{\mathbb{V}[\hat{\mathbf{x}}_m(t)-\mathbf{x}_m(t)]} \\ 
  {\text{s.t.}}&{\sum_{i=t-w}^{t-1} \lambda_i=1} 
\end{array}
\end{equation}
which can be solved with the method of Lagrange multipliers. 
\begin{thm}
\revadd{Given a} pre-trained covariance function $K(\mathbf{x},\mathbf{x}')$, the linear unbiased estimator of Gaussian process regression can estimate $\mathbf{p}_m(t)$ in polynomial time.
\end{thm}
\begin{proof}
Ordinary Gaussian process regression uses a linear unbiased estimator for $\mathbf{x}_m(t)$. We can use Lagrange multipliers to extract the $\lambda_i$ parameters from the optimization problem. 
\begin{align*}
L(\boldsymbol{\lambda},\mu)&=\mathbb{V}[\hat{\mathbf{x}}_m(t)-\mathbf{x}_m(t)] +\mu(\sum_{i=t-w}^{t-1} \lambda_i-1)\\
&=\mathbb{E}[\sum_{i=t-w}^{t-1} \lambda_i \mathbf{x}_m(i)-\mathbf{x}_m(t)]^{2}+\mu(\sum_{i=t-w}^{t-1}\lambda_{i}-1)\\
&=\sum_{i=t-w}^{t-1}\lambda_{i}\mathbb{E}[\mathbf{x}_m(i)-\mathbf{x}_m(t)]^{2}\\
&\;-\frac{1}{2}\sum_{i,j}\lambda_{i}\lambda_{j}\mathbb{E}[\mathbf{x}_m(i)-\mathbf{x}_m(j)]^2+\mu(\sum_{i=t-w}^{t-1}\lambda_{i}-1)
\end{align*}
where $\mathbb{E}[\mathbf{x}_m(i)-\mathbf{x}_m(t)]^{2}$ and $\mathbb{E}[\mathbf{x}_m(i)-\mathbf{x}_m(j)]^2$ are calculated from the pre-trained covariance function $K(\mathbf{x},\mathbf{x}')$. Then, we take partial derivatives of $L(\boldsymbol{\lambda},\mu)$ and set them to 0:
\begin{align}
\frac{\partial L(\boldsymbol{\lambda},\mu)}{\partial \boldsymbol{\lambda}}=0\\
\frac{\partial L(\boldsymbol{\lambda},\mu)}{\partial \mu}=0
\end{align}
obtaining a system of linear equations. There exist several algorithms for solving it, such as Gaussian elimination with $\mathcal{O}(w^3)$ computation complexity.
\end{proof}

As the prediction is a distribution for each $t$, we have the confidence intervals $\mathcal{I}_m(t)$, showing the probability of suffering a \ac{gnss} spoofing attack. As \ac{gnss} and network locations are with observational noises, the interval follows a Gaussian function at each time $t$ for each source. Then, its mean $\hat{\mathbf{p}}_m(t), m \in \{0,1,...,M\}$ and standard deviation $\boldsymbol{\sigma}_m (t)$ are sufficient to depict the interval as
\begin{equation}
    \mathcal{I}_m(t) \sim \mathcal{N} (\hat{\mathbf{p}}_m(t), \boldsymbol {\Sigma }_m (t)),m=0,1,...,M
\label{eq:conint}
\end{equation}
where $\boldsymbol {\Sigma }_m (t) = \operatorname {diag} ([\boldsymbol {\sigma }_m (t)]^{\circ 2})$ is a two-by-two diagonal matrix and $\circ 2$ is the \emph{Hadamard power}. The result of an example is demonstrated as Fig.~\ref{fig:gp_part}. The lines through the location points are the estimated means and the color bands are the estimated variances. 
\begin{figure}
\begin{centering}
\includegraphics[width=\columnwidth]{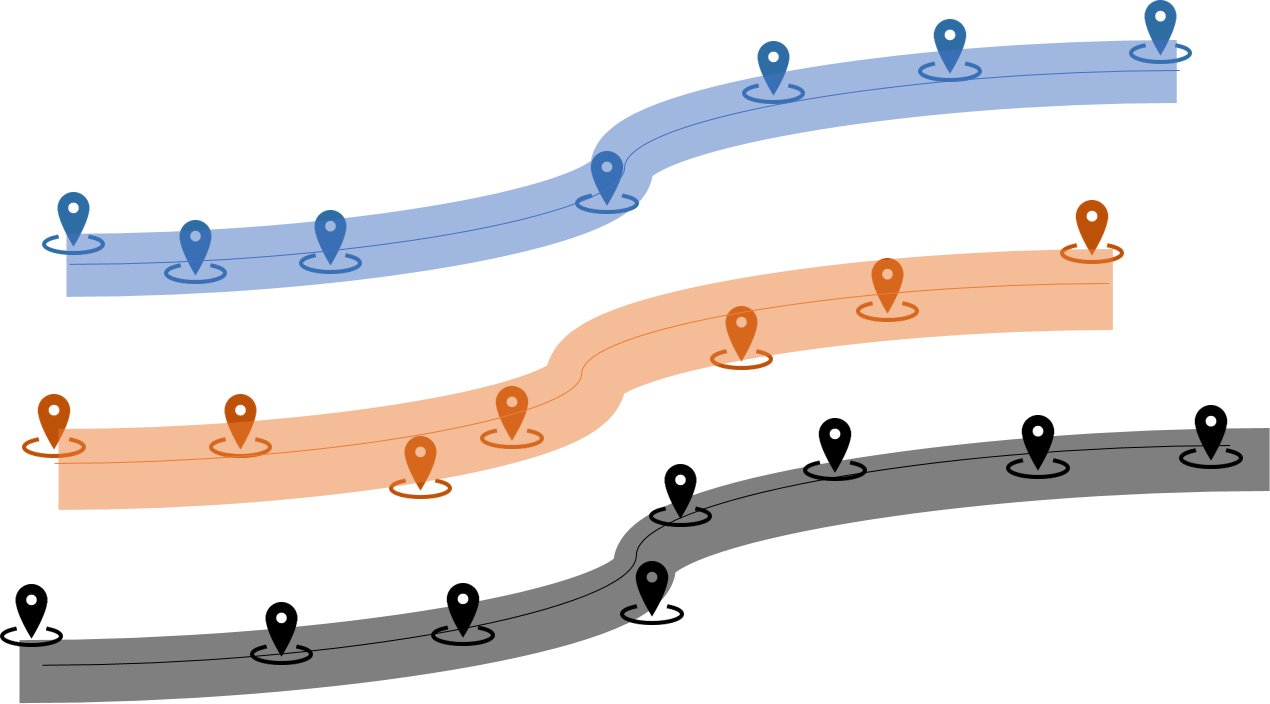}
\par\end{centering}
\caption{Gaussian process for modeling residual part of estimated positions (tiled view).}
\label{fig:gp_part}
\end{figure}

\subsection{Decision-Making using Likelihood}
\label{croseq}
The third component, decision-making, fuses confidence intervals of all locations (i.e., \ac{gnss}, Wi-Fi, and cellular) into one likelihood function, and then uses the \ac{npl} to maximize the true positive rate (of whether current \ac{gnss} position is under attack).
\subsubsection{Time Perspective}
For the data sequence, $S$, and its corresponding confidence intervals $\mathcal{I}_m(t)=\hat{\mathbf{p}}_m(t)+\mathbf{x}_m(t)$ from Algorithm \ref{alg:confiden}, we need to sequentially combine confidence intervals into a distribution of time $t$ for infrastructure $m$. The main idea is to integrate the weighted confidence intervals over $t$ with weights $K(m,t)$ and the combined distribution denoted as $Z(m,t)$. Assume that $\mathcal{I}_m(t)$ are independent random variables. Note that the time slot is from $t-w$ to $t$, and $K(m,t)$ are from a kernel function, such that the summation of $K(m,t)$ from $t-w$ to $t$ is 1. Then, the distribution $Z(m,t)$ is calculated by using the moment-generating function: 
\begin{equation}
    M_{\mathcal{I}_m(t)}(s)=\mathbb{E}[e^{s\mathcal{I}_m(t)}]
\end{equation}
Considering a weighted integral over $t$, 
\begin{equation}
    Z(m,t)=\int_{t-m}^{t}K(m,t')\mathcal{I}_m(t')\textrm{d}t'
\end{equation}
which is used in a discrete form in practice, 
\begin{equation}
    Z(m,t)=\sum_{t'=t-w}^{t}K(m,t')\mathcal{I}_m(t')
\label{eq:z_mt}
\end{equation}
and leads to its moment-generating function:
\begin{align*}
    M_{Z\left( m,t \right)}\left( s \right) &=\mathbb{E}\left[ e^{sZ\left( m,t \right)} \right]\\
    &=\mathbb{E}\left[ e^{s\sum_{t'=t-w}^t{K}\left( m,t' \right) \mathcal{I}_m(t')} \right] \\
    &=\prod_{t'=t-w}^t{\mathbb{E}}\left[ e^{sK\left( m,t' \right) \mathcal{I}_m(t')} \right] \\
    &=\prod_{t'=t-w}^t{M}_{\mathcal{I}_m(t')}\left( K\left( m,t' \right) s \right) 
\end{align*}
Recall that the moment-generating function of a Normal distribution, $\mathcal{N}(\mu ,\sigma ^{2})$, is ${\displaystyle \exp({s\mu +{\frac {1}{2}}\sigma ^{2}s^{2}}})$. Thus, 
\begin{align}
&M_{Z\left( m,t \right)}\left( s \right) \nonumber \\
&=\prod_{t'=t-w}^t{e^{-K\left( m,t' \right) \boldsymbol{\mu }^m\left( t' \right) +\left[ K\left( m,t' \right) \right] ^2 [\sigma ^m\left( t' \right)]^{\circ 2} /2}} \nonumber \\
&=e^{-\sum_{t'=t-w}^t{K\left( m,t' \right) \boldsymbol{\mu }^m\left( t' \right)}+\sum_{t'=t-w}^t{\left[ K\left( m,t' \right) \right] ^2 [\sigma ^m\left( t' \right)]^{\circ 2} /2}}.
\end{align}
$Z(m,t)$ follows a normal distribution, $\mathcal{N}(\sum_{t'=t-w}^t{K\left( m,t' \right) \boldsymbol{\mu }^m\left( t' \right)}, \sum_{t'=t-w}^t{\left[ K\left( m,t' \right) \right] ^2 \boldsymbol{\Sigma} ^m\left( t' \right)})$, which means we can compute the sequence-wise combined distribution $Z(m,t)$ by taking the weighted mean of distributions of time slots.
\subsubsection{Category Perspective}
For the combined distribution $Z(m,t)$ for all categories, the $m$th test statistics for two hypotheses are 
\revadd{
\begin{align*}
    \bullet\; \varLambda_m&\left( \mathbf{p}_0 \left( t \right) \right) |H_0
    =\frac{1}{\sigma\left(Z\left(m,t\right)\right)}\hfill\\
    &\times\varphi\left(\frac{\sum_{t'=t-w}^{t}K(m,t')\mathbf{p}_{0}(t')-\mu\left(Z\left(m,t\right)\right)}{\sigma\left(Z\left(m,t\right)\right)}\right)\\
    \bullet\; \varLambda_m&\left( \mathbf{p}_0 \left( t \right) \right) |H_1
    =\frac{1}{\sigma\left(Z\left(m,t\right)\right)}\hfill\\
    &\times\varphi\left(\frac{\sum_{t'=t-w}^{t}K(m,t')\mathbf{p}_{0}(t')-\mu\left(Z\left(m,t\right)\right)-\boldsymbol{\delta}_\text{d}}{\sigma\left(Z\left(m,t\right)\right)}\right)
\end{align*}}where $\varphi\left(\cdot\right)$ is the standard normal distribution and $\boldsymbol{\delta}_\text{d}$ is a non-centrality parameter because of \ac{gnss} spoofing. For $M$ kinds of locations, the fused test statistics is 
\begin{equation}
    \varLambda_{1:M}\left( \mathbf{p}_0 \left( t \right) \right)=\prod_{m=0}^M \varLambda_m\left( \mathbf{p}_0 \left( t \right) \right) |H_0 \;.
\label{eq:fused_test}
\end{equation}
We also define the likelihood of \ac{gnss} being under attack as $f_0\left( \mathbf{p}_0 \left( t \right) \right) \triangleq \varLambda_{1:M}\left( \mathbf{p}_0 \left( t \right) \right)$. To simplify the computation, we observe that $f_0\left( \mathbf{p}_0 \left( t \right) \right)$ is also a Gaussian function, $S\, \mathcal N (\mu,\sigma^2)$, where
\begin{equation}
    \sigma = \left(\sum_{i=0}^M {\sigma\left(Z\left(m,t\right)\right)}^{-2} \right)^{-\frac{1}{2}}
\label{eq:varcom}
\end{equation}
\begin{equation}
    \mu = \sigma^2 \sum_{i=0}^M {\sigma\left(Z\left(m,t\right)\right)}^{-2} \mu\left(Z\left(m,t\right)\right)
\label{eq:poscom}
\end{equation}
\begin{align}
    S =& \frac{(2\pi)^{-\frac{M}{2}} \sigma e^{({\mu^2}/{\sigma^{2}} - \sum_{i=0}^M  {\mu\left(Z\left(m,t\right)\right)^2}/{{\sigma\left(Z\left(m,t\right)\right)}^{2}} )/2}}{\prod_{m=0}^M\sigma\left(Z\left(m,t\right)\right)}
\end{align}
(the proof is omitted here due to space limitations). To reduce the computation and avoid the joint probability tending to 0 when the dataset is large, we take a logarithm. Then, in the next part of the proposed scheme, the \ac{npl}-based algorithm will look for a threshold for decision-making.
\subsubsection{\ac{npl} for Decision-Making} This is shown as Algorithm \ref{alg:decision}. We denote the tolerable upper bound of the Type I error (false positive) probability as $P_{\text{FP}_\text{max}}$ and the detection threshold for $\varLambda_{1:M}(\mathbf{p}_0(t))$ as $\gamma$. Then, the optimization problem for \ac{gnss} spoofing detection is phrased as 
$$
\begin{matrix}
	\underset{\gamma}{\max}&		P\left( \log \varLambda_{1:M} \left( \mathbf{p}_0 \left( t \right) \right) \le \gamma \ |\ H_1 \right)\\
	\text{s.t}.&		P\left( \log \varLambda_{1:M} \left( \mathbf{p}_0 \left( t \right) \right) \le \gamma \ |\ H_0 \right) \le P_{\text{FP}_{\max}}\\
\end{matrix}.
$$
Through the optimization problem, we can find a proper threshold $\gamma$; when $\varLambda_{1:M} \left( \mathbf{p}_0 \left( t \right) \right) \le \gamma$ the decision is to raise alarm that \ac{gnss} is spoofed.  
\begin{algorithm}
\hspace*{\algorithmicindent} \textbf{Input} \textit{CI}\\
\hspace*{\algorithmicindent} \textbf{Parameter} $\gamma$\\
\hspace*{\algorithmicindent} \textbf{Output} \textit{IsAttack}
\begin{algorithmic}[1]
\State $Z(m,t) \gets$ Eq.~\eqref{eq:z_mt} \Comment{ Temporal perspective}
\State $\varLambda_{1:M} \left( \mathbf{p}_0 \left( t \right) \right) \gets$ Eq.~\eqref{eq:fused_test} \Comment{Categorical perspective}
\If{$\log \varLambda_{1:M} \left( \mathbf{p}_0 \left( t \right) \right) \le \gamma$} 
    \State \textit{IsAttack} $\gets$ \textit{True}
\Else
    \State \textit{IsAttack} $\gets$ \textit{False}
\EndIf 
\end{algorithmic}
\caption{Decision based on series from multiple opportunistic information sources \label{alg:decision}}
\end{algorithm}

\section{Numerical Results}
\label{numres}
We consider four methods as baselines to compare with our approach. We choose three metrics, true positive probability, $P_\text{TP}$, detection time delay, $\Delta T$, and the \ac{mae} of alternative (non-\ac{gnss}) position, $\mu$.

We start with three datasets, Datasets A and B in \cite{Baidu2019,JeoChoShiRoh:J19} respectively contain 6 real-world \ac{gnss} traces and were collected in cities. Dataset C \cite{OliSciIbrDip:J22} has 10 different traces in an urban area too, slightly different from A and B. The format is summarized in Table~\ref{tab:dataset}. The lengths of the traces range from about ten to tens of kilometers. The vehicle speed ranges from 0 km/h to 90 km/h. The positioning error of benign \ac{gnss}-provided positions is Gaussian with zero mean and variance of 0.9 meters. 

Cellular and Wi-Fi location data is generated from a simulator\footnote{The simulation parameters for the free-space path loss model are from \ac{lte} TR36.814 and 802.11n 2.4 GHz.}, which generates received signal strength from 4 \acpl{bs} or \acpl{ap}, respectively. \revadd{In the absence of \ac{ap} and \ac{bs} location information, those are randomly generated with an approximate 50 meters distance to the \ac{gnss} receiver trace. The \ac{ap}-based or \ac{bs}-based positioning algorithm is \ac{wcl} \cite{OliSciIbrDip:J22} and the resultant positions are observed with noise variance 33 or 9 in meters as per \cite{KuuFalKatDia:J18}.} The unavailability probability $U_m$ for networks is 0.05 for all $m=1,2,...,M$ and it is binomially distributed. 

The \ac{gnss} location data when under attack are generated based on \cite{SheWonCheChe:C20}, updated once per second (1 Hz). The attack strategy \revadd{makes} the lateral deviation (distance of the vehicle sideways shift on a road) between the position of the \ac{gnss} victim accepted and the real position as large as possible. The attack consists of two stages: (i) vulnerability profiling: the attacker performs a constant spoofing to make a small constant deviation from the real position, and (ii) aggressive spoofing: after the victim accepted the spoofed position, the attacker makes the deviation grow exponentially. 

\revadd{By incorporating the synthesized network infrastructure locations, the resultant device position estimates, and then the attack-induced deviations into the original three datasets, we get augmented Datasets A, B, and C for our simulation experiments.}

\begin{table}
\centering
\caption{Format of \ac{gnss} mobile platform datasets.}
\begin{tabular}{lrrrr}
\hline
\hline
       &      Sample Rate   &   Dataset A   &   Dataset B   &   Dataset C \\
\hline
Timestamp     & 200Hz & \checkmark & \checkmark & \checkmark \\
Actual Location & 1Hz & \checkmark & \checkmark & \checkmark \\
GNSS Location & 1Hz & \checkmark & \checkmark & \checkmark \\
On-board Sensors & 200Hz & \checkmark & \checkmark & \text{\sffamily X} \\
\hline
\hline
\end{tabular}
\label{tab:dataset}
\end{table}

The used machine is HP-EliteDesk-800-G2-TWR with 32 GB memory and 3.40 GHz CPU. The operating system is Ubuntu 20.04.3 LTS 64-bit and the programming environment is Python 3.8.10 64-bit. Then, in our algorithm, the Gaussian process implementation uses the Python library runlmc, and the convex optimization part is based on the Python library CVXPY 1.1. 

\subsection{Baseline Methods}
\subsubsection{Signals of Opportunity}
\cite{OliSciIbrDip:J22} uses the broadcast signals from the network \acpl{bs} to validate \ac{gps}-provided position. It assumes that \ac{bs} positions are available and uses the received signal strength, $RSS_i(t)$, to estimate the distance between the mobile platform and the \ac{bs}; $i$ is the \ac{bs} index and $t$ is the time of the received signal. Based on the signal strength, we compute weights $\mathbf{w}=[w_1,w_2,...,w_N]$ and the estimated mobile platform position as the weighted centroid $Y_{est}=\frac{\mathbf{w} \cdot \mathbf{p}_{bs}}{|\mathbf{w}|}$, where $\mathbf{p}_{bs}$ is concatenated positions of all \acpl{bs}. If the distance of $Y_{est}$ and the \ac{gps}-provided position is higher than a threshold, the protocol raises an alarm (spoofing). 
\subsubsection{Kalman Filter}
The extended Kalman filter (EKF) fuses \ac{imu} and \ac{gnss} measurements. We estimate the position of the mobile platform; the state of the system refers to the motion of the mobile platform. The state estimation problem is expressed as 
\begin{equation}
    \begin{array}{*{20}{l}}
    {\mathbf{p}(t)=f(\mathbf{p}(t),\mathbf{u}(t),\mathbf{w}(t))}\\
    {\mathbf{y}(t)=g(\mathbf{p}(t),\mathbf{n}(t))}
    \end{array}
\end{equation}
where $\mathbf{p}(t)$ is location, $f$ is the motion equation, $\mathbf{u}$ is input, $\mathbf{w}$ is input noise, $g$ is observation equation, and $\mathbf{n}$ is observation noise. As the distribution does not remain Gaussian after the nonlinear transformation, $f$ and the noise are approximated as Gaussian. The Kalman filter minimizes the error of observation and motion, so we can recursively get the mean and covariance matrix of position $\mathbf{p}(t)$ \cite{CecForLauTom:J21}. 
\subsubsection{Combined Metrics}
In \cite{RotCheLoWal:J21}, multiple detection metrics, such as the received power, autocorrelation distortion, pseudoranges, carrier phase differences, and direction of arrival, are used. These $M$ metrics are considered statistically independent and a likelihood ratio function is used to combine them:
\begin{equation}
    \log \varLambda _{1:M} = \sum_{m=1}^M \log \varLambda _m.
\end{equation}
Then, the false negative probability is minimized by following the \ac{npl} paradigm.
\subsubsection{Particle Filter}
It is based on the Markov Monte Carlo method, with its central concept centered around sample generation through a stochastic process. It does not assume the location follows a Gaussian distribution. The implementation has four steps: Step 1 uniformly generates $L$ particles, $\mathbf{p}^l_m(t),l=1,2,...,L$, of locations around the initial location. Step 2 calculates the error, $e^l_m(t)$, between particles and the localization module data and uses the error to update $w^l_m(t)$. Then, the estimated position is $\hat{\mathbf{p}}_m(t)=\sum_{l=1}^L w^l_m(t) \mathbf{p}^l_m(t) / \sum_{l=1}^L w^l_m(t)$. Step 3 resampling avoids particle degenerating and removes particles with weights less than $1/\sum_{l=1}^L (w^l_m(t))^2$. Step 4 involves a recursive update, updating $t=t+1$, and then returning to Step 2. 
\subsection{Detection Error Probability}


\begin{figure*}
    \centering
    \begin{subfigure}[b]{0.24\textwidth}
    \begin{tikzpicture}[scale=.5,font=\Large]
    \begin{axis}[
        xlabel={Attack deviation [m]},
        ylabel={True positive rate [\%]},
        xmin=1, xmax=10,
        ymin=0, ymax=100,
        xtick={1,2,3,4,5,6,7,8,9,10},
        legend cell align={left},
        legend pos=south east,
        legend columns=1, 
        xmajorgrids=true,
        ymajorgrids=true,
        grid style=dashed,
    ]
    \addplot[
        color=mycolor1,
        mark=square,
        ]
        coordinates {
        (1,60)(2,72)(3,77.5)(4,83)(5,86)(6,87.5)(7,90.8)(8,92.8)(9,96.2)(10,96.5)
        };
        \addlegendentry{$P_{\text{FP}_{\max}}=0.05$}
     \addplot[
        color=mycolor2,
        mark=x,
        ]
        coordinates {
        (1,63)(2,75.5)(3,81)(4,87)(5,90)(6,91.5)(7,93.5)(8,93.4)(9,96.5)(10,96.5)
        };
        \addlegendentry{$P_{\text{FP}_{\max}}=0.1$}
     \addplot[
        color=mycolor3,
        mark=triangle,
        ]
        coordinates {
        (1,67)(2,79)(3,85)(4,90.5)(5,93.5)(6,94.7)(7,96.3)(8,96.5)(9,96.8)(10,96.8)
        };
        \addlegendentry{$P_{\text{FP}_{\max}}=0.15$}

    \addplot[color=mycolor1,mark=square,dashed]
        coordinates {(1,39)(2,46)(3,52)(4,56)(5,64)(6,73)(7,77)(8,89)(9,93)(10,93)};
    \addplot[color=mycolor2,mark=x,dashed]
        coordinates {(1,52.5)(2,60)(3,63.5)(4,70)(5,82)(6,88)(7,90)(8,96)(9,96)(10,96)};
    \addplot[color=mycolor3,mark=triangle,dashed]
        coordinates {(1,61)(2,70)(3,73.5)(4,80)(5,89)(6,96)(7,97.5)(8,99.5)(9,99.5)(10,99.5)};
    \end{axis}
    \end{tikzpicture}
    \caption{PDS with networks}
    \end{subfigure}
    \begin{subfigure}[b]{0.24\textwidth}
    \begin{tikzpicture}[scale=.5,font=\Large]
    \begin{axis}[
        xlabel={Attack deviation [m]},
        ylabel={True positive rate [\%]},
        xmin=1, xmax=10,
        ymin=0, ymax=100,
        xtick={1,2,3,4,5,6,7,8,9,10},
        legend cell align={left},
        legend pos=south east,
        legend columns=1, 
        xmajorgrids=true,
        ymajorgrids=true,
        grid style=dashed,
    ]
    \addplot[
        color=mycolor1,
        mark=square,
        ]
        coordinates {
        (1,50.5)(2,61)(3,66)(4,71.5)(5,75)(6,78)(7,81)(8,82.5)(9,87.5)(10,90)
        };
        \addlegendentry{$P_{\text{FP}_{\max}}=0.05$}
    \addplot[
        color=mycolor2,
        mark=x,
        ]
        coordinates {
        (1,57.5)(2,67.5)(3,72.5)(4,77.5)(5,81)(6,83.5)(7,85)(8,87)(9,92)(10,93)
        };
        \addlegendentry{$P_{\text{FP}_{\max}}=0.1$}
    \addplot[
        color=mycolor3,
        mark=triangle,
        ]
        coordinates {
        (1,62.2)(2,72)(3,77.5)(4,82)(5,85)(6,87)(7,87)(8,89)(9,94)(10,95)
        };
        \addlegendentry{$P_{\text{FP}_{\max}}=0.15$}
    
    \addplot[color=mycolor1,mark=square,dashed]
        coordinates {(1,38)(2,45)(3,50)(4,55)(5,66)(6,71)(7,73)(8,78)(9,83)(10,88)};
    \addplot[color=mycolor2,mark=x,dashed]
        coordinates {(1,47.5)(2,58)(3,60)(4,65)(5,72)(6,78)(7,78)(8,84)(9,87)(10,93)};
    \addplot[color=mycolor3,mark=triangle,dashed]
        coordinates {(1,53)(2,63)(3,65)(4,68)(5,77)(6,80)(7,80)(8,87)(9,90)(10,95)};
    \end{axis}
    \end{tikzpicture}
    \caption{Signals of opportunity}
    \end{subfigure}
    \begin{subfigure}[b]{0.24\textwidth}
    \begin{tikzpicture}[scale=.5,font=\Large]
    \begin{axis}[
        xlabel={Attack deviation [m]},
        ylabel={True positive rate [\%]},
        xmin=1, xmax=10,
        ymin=0, ymax=100,
        xtick={1,2,3,4,5,6,7,8,9,10},
        legend cell align={left},
        legend pos=south east,
        legend columns=1, 
        xmajorgrids=true,
        ymajorgrids=true,
        grid style=dashed,
    ]
    \addplot[color=mycolor6,mark=square]
        coordinates {(1,62)(2,74)(3,79.5)(4,85)(5,87)(6,89)(7,91.5)(8,91.8)(9,92.5)(10,93)};
        \addlegendentry{Proposed with Wi-Fi}
    \addplot[color=mycolor4,mark=o]
        coordinates {(1,57)(2,69)(3,73.5)(4,79)(5,82.5)(6,86)(7,86.5)(8,86.7)(9,87)(10,87.5)};
        \addlegendentry{Signals of opportunity}
    \end{axis}
    \end{tikzpicture}
    \caption{Only Wi-Fi, $P_{\text{FP}_{\max}}=0.1$}
    \end{subfigure}
    \begin{subfigure}[b]{0.24\textwidth}
    \begin{tikzpicture}[scale=.5,font=\Large]
    \begin{axis}[
        xlabel={Attack deviation [m]},
        ylabel={True positive rate [\%]},
        xmin=1, xmax=10,
        ymin=0, ymax=100,
        xtick={1,2,3,4,5,6,7,8,9,10},
        legend cell align={left},
        legend pos=south east,
        legend columns=1, 
        xmajorgrids=true,
        ymajorgrids=true,
        grid style=dashed,
    ]
    \addplot[color=mycolor6,mark=square]
        coordinates {(1,60)(2,74)(3,80)(4,85.5)(5,89)(6,91)(7,92)(8,93)(9,95)(10,95.5)};
        \addlegendentry{Proposed with cellular}
    \addplot[color=mycolor4,mark=o]
        coordinates {(1,57)(2,68)(3,73)(4,77)(5,81)(6,83)(7,85)(8,86)(9,89)(10,90)};
        \addlegendentry{Signals of opportunity}
    \end{axis}
    \end{tikzpicture}
    \caption{Only cellular, $P_{\text{FP}_{\max}}=0.1$}
    \end{subfigure}
    \caption{(Case 1) $P_\text{TP}$ of PDS using network-based positioning and signals of opportunity method \cite{OliSciIbrDip:J22} as a function of the spoofing deviation. Solid lines are for \revadd{augmented} Datasets A and B, and dashed lines are for augmented Dataset C.}
    \label{fig:tpdevnet}
\end{figure*}

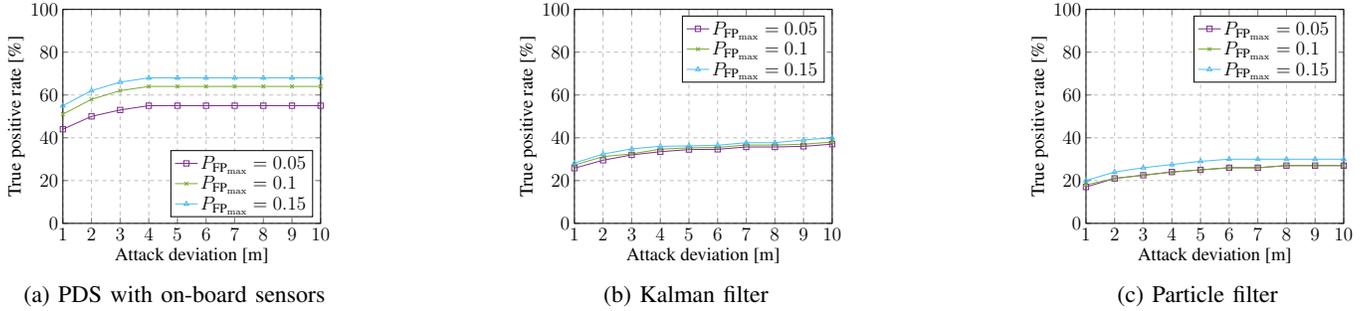
\begin{figure*}
    \centering
    \begin{subfigure}[b]{0.25\textwidth}
    \begin{tikzpicture}[scale=.5,font=\Large]
    \begin{axis}[
        xlabel={Attack deviation [m]},
        ylabel={True positive rate [\%]},
        xmin=1, xmax=10,
        ymin=0, ymax=100,
        xtick={1,2,3,4,5,6,7,8,9,10},
        legend cell align={left},
        legend pos=south east,
        legend columns=1, 
        xmajorgrids=true,
        ymajorgrids=true,
        grid style=dashed,
    ]
    \addplot[
        color=mycolor1,
        mark=square,
        ]
        coordinates {
        (1,44)(2,50)(3,53)(4,55)(5,55)(6,55)(7,55)(8,55)(9,55)(10,55)
        };
        \addlegendentry{$P_{\text{FP}_{\max}}=0.05$}
     \addplot[
        color=mycolor2,
        mark=x,
        ]
        coordinates {
        (1,51)(2,58)(3,62)(4,64)(5,64)(6,64)(7,64)(8,64)(9,64)(10,64)
        };
        \addlegendentry{$P_{\text{FP}_{\max}}=0.1$}
     \addplot[
        color=mycolor3,
        mark=triangle,
        ]
        coordinates {
        (1,55)(2,62)(3,66)(4,68)(5,68)(6,68)(7,68)(8,68)(9,68)(10,68)
        };
        \addlegendentry{$P_{\text{FP}_{\max}}=0.15$}
    \end{axis}
    \end{tikzpicture}
    \caption{PDS with on-board sensors}
    \end{subfigure}
    \hfill
    \begin{subfigure}[b]{0.25\textwidth}
    \begin{tikzpicture}[scale=.5,font=\Large]
    \begin{axis}[
        xlabel={Attack deviation [m]},
        ylabel={True positive rate [\%]},
        xmin=1, xmax=10,
        ymin=0, ymax=100,
        xtick={1,2,3,4,5,6,7,8,9,10},
        legend cell align={left},
        legend pos=north east,
        legend columns=1, 
        xmajorgrids=true,
        ymajorgrids=true,
        grid style=dashed,
    ]
    \addplot[
        color=mycolor1,
        mark=square,
        ]
        coordinates {
        (1,25.8)(2,29.5)(3,32)(4,33.5)(5,34.5)(6,34.6)(7,35.7)(8,35.7)(9,36)(10,37)
        };
        \addlegendentry{$P_{\text{FP}_{\max}}=0.05$}
     \addplot[
        color=mycolor2,
        mark=x,
        ]
        coordinates {
        (1,27.2)(2,31.2)(3,32.5)(4,34.7)(5,35.3)(6,35.5)(7,36.6)(8,36.6)(9,37)(10,38)
        };
        \addlegendentry{$P_{\text{FP}_{\max}}=0.1$}
     \addplot[
        color=mycolor3,
        mark=triangle,
        ]
        coordinates {
        (1,28.2)(2,32.4)(3,34.8)(4,36)(5,36.2)(6,36.5)(7,37.7)(8,37.7)(9,39)(10,40)
        };
        \addlegendentry{$P_{\text{FP}_{\max}}=0.15$}
        \end{axis}
    \end{tikzpicture}
    \caption{Kalman filter}
    \end{subfigure}
    \hfill
    \begin{subfigure}[b]{0.25\textwidth}
    \begin{tikzpicture}[scale=.5,font=\Large]
    \begin{axis}[
        xlabel={Attack deviation [m]},
        ylabel={True positive rate [\%]},
        xmin=1, xmax=10,
        ymin=0, ymax=100,
        xtick={1,2,3,4,5,6,7,8,9,10},
        legend cell align={left},
        legend pos=north east,
        legend columns=1, 
        xmajorgrids=true,
        ymajorgrids=true,
        grid style=dashed,
    ]
    \addplot[
        color=mycolor1,
        mark=square,
        ]
        coordinates {
        (1,17)(2,21)(3,22.5)(4,24)(5,25)(6,26)(7,26)(8,27)(9,27)(10,27)
        };
        \addlegendentry{$P_{\text{FP}_{\max}}=0.05$}
     \addplot[
        color=mycolor2,
        mark=x,
        ]
        coordinates {
        (1,18)(2,21)(3,22.5)(4,24)(5,25)(6,26)(7,26)(8,27)(9,27)(10,27)
        };
        \addlegendentry{$P_{\text{FP}_{\max}}=0.1$}
     \addplot[
        color=mycolor3,
        mark=triangle,
        ]
        coordinates {
        (1,20)(2,24)(3,26)(4,27.5)(5,29)(6,30)(7,30)(8,30)(9,30)(10,30)
        };
        \addlegendentry{$P_{\text{FP}_{\max}}=0.15$}
    \end{axis}
    \end{tikzpicture}
    \caption{Particle filter}
    \end{subfigure}
    \caption{(Case 2) $P_\text{TP}$ of PDS only using on-board sensors, Kalman filter, and particle filter as a function of the different spoofing deviation for augmented Datasets A and B.}
    \label{fig:tpdevmot}
\end{figure*}

\begin{figure}
    \centering
    \begin{subfigure}[b]{0.24\textwidth}
    \begin{tikzpicture}[scale=.5,font=\Large]
    \begin{axis}[
        xlabel={Attack deviation [m]},
        ylabel={True positive rate [\%]},
        xmin=1, xmax=10,
        ymin=0, ymax=100,
        xtick={1,2,3,4,5,6,7,8,9,10},
        legend cell align={left},
        legend pos=south east,
        legend columns=1, 
        xmajorgrids=true,
        ymajorgrids=true,
        grid style=dashed,
    ]
    \addplot[
        color=mycolor1,
        mark=square,
        ]
        coordinates {
        (1,60.3)(2,71.6)(3,77.7)(4,83.1)(5,86.4)(6,87.9)(7,90.7)(8,93.4)(9,95.8)(10,96.9)
        };
        \addlegendentry{$P_{\text{FP}_{\max}}=0.05$}
     \addplot[
        color=mycolor2,
        mark=x,
        ]
        coordinates {
        (1,63)(2,75.5)(3,81)(4,87)(5,90)(6,91.5)(7,93.5)(8,93.7)(9,97)(10,97)
        };
        \addlegendentry{$P_{\text{FP}_{\max}}=0.1$}
     \addplot[
        color=mycolor3,
        mark=triangle,
        ]
        coordinates {
        (1,67)(2,79)(3,85)(4,90.5)(5,93.8)(6,95)(7,96.5)(8,97)(9,97.2)(10,97.8)
        };
        \addlegendentry{$P_{\text{FP}_{\max}}=0.15$}
    \end{axis}
    \end{tikzpicture}
    \caption{PDS with all sources}
    \end{subfigure}
    \begin{subfigure}[b]{0.24\textwidth}
    \begin{tikzpicture}[scale=.5,font=\Large]
    \begin{axis}[
        xlabel={Attack deviation [m]},
        ylabel={True positive rate [\%]},
        xmin=1, xmax=10,
        ymin=0, ymax=100,
        xtick={1,2,3,4,5,6,7,8,9,10},
        legend cell align={left},
        legend pos=south east,
        legend columns=1, 
        xmajorgrids=true,
        ymajorgrids=true,
        grid style=dashed,
    ]
    \addplot[
        color=mycolor1,
        mark=square,
        ]
        coordinates {
        (1,45)(2,55)(3,60)(4,65)(5,69)(6,70)(7,73)(8,74.5)(9,78)(10,80)
        };
        \addlegendentry{$P_{\text{FP}_{\max}}=0.05$}
     \addplot[
        color=mycolor2,
        mark=x,
        ]
        coordinates {
        (1,52)(2,60)(3,65)(4,69.5)(5,73)(6,76)(7,78)(8,80.5)(9,85)(10,90)
        };
        \addlegendentry{$P_{\text{FP}_{\max}}=0.1$}
     \addplot[
        color=mycolor3,
        mark=triangle,
        ]
        coordinates {
        (1,57)(2,64)(3,68)(4,72.5)(5,75)(6,78)(7,81.5)(8,83)(9,86)(10,92)
        };
        \addlegendentry{$P_{\text{FP}_{\max}}=0.15$}
    \end{axis}
    \end{tikzpicture}
    \caption{Combined metrics}
    \end{subfigure}
    \caption{(Case 3) $P_\text{TP}$ of PDS using all sources of information and the combined metrics algorithm, as a function of the different spoofing deviation for augmented Datasets A and B.}
    \label{fig:tpdevall}
\end{figure}
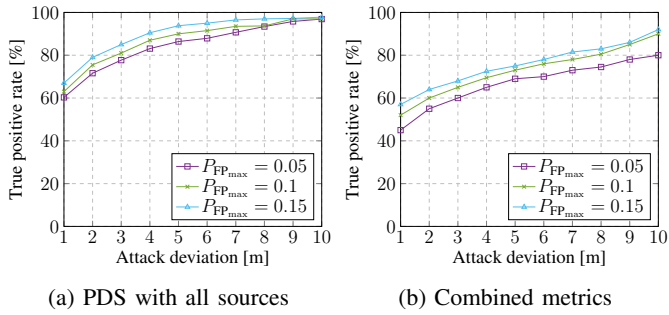

To fairly compare PDS to the baseline methods, we consider detection error probability \revadd{of (1)} network interfaces only, (2) on-board sensors only, and (3) all sources. Then, we investigate the true positive when fixing the false positive rate at 5,10, and 15\%. \revadd{Using the fixed false positive rate and historical data, the threshold $\gamma$ is selected from 0.25 to 10.} The sliding window size is set to $w=20$ seconds. We choose the \ac{rbf} kernel $K(r)=\exp(-r^2)$ as the kernel function. 

For case 1, the baseline network-based scheme \cite{OliSciIbrDip:J22} and PDS have the same network connections. Fig.~\ref{fig:tpdevnet} shows the true positive rate versus attack-induced deviation (from 1 to 10 meters). PDS here improves modestly, around 7\%, over \cite{OliSciIbrDip:J22}; it can reach 93--97\% true positive rates when the attack-induced deviation distance is 5--10 meters, vs. 85--93\% (network-based algorithm). Both methods can detect most attacks and thus resist dedicated designed spoofing as per \cite{SheWonCheChe:C20}. 

For case 2, the Kalman and particle filter-based approaches use the same input data as PDS, i.e., speed, acceleration, rotation, and \ac{gnss} position. Fig.~\ref{fig:tpdevmot} shows the true positive rate versus attack-induced deviation for a given false alarm rate. The Kalman and the particle filter have similar performance. PDS has 15--30\% true positive gain. However, overall, information from on-board sensors in the dataset does not resist dedicated designed spoofing \cite{SheWonCheChe:C20}, so all three methods can not achieve high performance (at most 70\%).

For case 3, the baseline combined metrics scheme and PDS use both on-board sensors and network interfaces. Fig.~\ref{fig:tpdevall} shows the true positive rate versus deviation from 1 to 10 meters. Due to the absence of motion uncertainty modeling of location data in the algorithm, the combined metrics scheme is not good at fusing heterogeneous data (i.e., location, speed, and acceleration) in our case, thus at most 20\% performance gain.

Considering cases 1 and 3, PDS has around 1.5\% performance gain after adding \ac{imu}. \ac{imu} is especially helpful if the growth rate of deviation is large. 

\subsection{Detection Time Delay}
\begin{figure*}
    \centering
    \begin{subfigure}[b]{0.24\textwidth}
    \begin{tikzpicture}[scale=.5,font=\Large]
    \begin{axis}[
        xlabel={Attack deviation [m]},
        ylabel={Detection time delay [s]},
        xmin=1, xmax=10,
        ymin=0, ymax=15,
        xtick={1,2,3,4,5,6,7,8,9,10},
        ytick={0,3,6,9,12,15},
        legend cell align={left},
        legend pos=north east,
        legend columns=1, 
        xmajorgrids=true,
        ymajorgrids=true,
        grid style=dashed,
    ]
    \addplot[
        color=mycolor1,
        mark=square,
        ]
        coordinates {
        (1,8)(2,3.7)(3,2.4)(4,0.8)(5,0.65)(6,0.5)(7,0)(8,0)(9,0)(10,0)
        };
        \addlegendentry{$P_{\text{FP}_{\max}}=0.05$}
     \addplot[
        color=mycolor2,
        mark=x,
        ]
        coordinates {
        (1,7.6)(2,3.8)(3,2.3)(4,0.8)(5,0.65)(6,0.35)(7,0)(8,0)(9,0)(10,0)
        };
        \addlegendentry{$P_{\text{FP}_{\max}}=0.1$}
     \addplot[
        color=mycolor3,
        mark=triangle,
        ]
        coordinates {
        (1,6)(2,3)(3,1.9)(4,0.75)(5,0.58)(6,0.35)(7,0)(8,0)(9,0)(10,0)
        };
        \addlegendentry{$P_{\text{FP}_{\max}}=0.15$}
        
    \addplot[color=mycolor1,mark=square,dashed]
        coordinates {(1,10)(2,8)(3,6.5)(4,5.5)(5,4.5)(6,3.5)(7,1.5)(8,0.3)(9,0)(10,0)};
    \addplot[color=mycolor2,mark=x,dashed]
        coordinates {(1,8)(2,5)(3,4)(4,3.5)(5,3)(6,2)(7,1)(8,0)(9,0)(10,0)};
    \addplot[color=mycolor3,mark=triangle,dashed]
        coordinates {(1,6)(2,3)(3,2)(4,1.6)(5,1)(6,0.5)(7,0.3)(8,0)(9,0)(10,0)};
    \end{axis}
    \end{tikzpicture}
    \caption{PDS with networks}
    \end{subfigure}
    \hfill
    \begin{subfigure}[b]{0.24\textwidth}
    \begin{tikzpicture}[scale=.5,font=\Large]
    \begin{axis}[
        xlabel={Attack deviation [m]},
        ylabel={Detection time delay [s]},
        xmin=1, xmax=10,
        ymin=0, ymax=15,
        xtick={1,2,3,4,5,6,7,8,9,10},
        ytick={0,3,6,9,12,15},
        legend cell align={left},
        legend pos=north east,
        legend columns=1, 
        xmajorgrids=true,
        ymajorgrids=true,
        grid style=dashed,
    ]
    \addplot[
        color=mycolor1,
        mark=square,
        ]
        coordinates {
        (1,8)(2,4.8)(3,3.5)(4,2)(5,1.4)(6,1)(7,0.55)(8,0.65)(9,0.65)(10,0.5)
        };
        \addlegendentry{$P_{\text{FP}_{\max}}=0.05$}
     \addplot[
        color=mycolor2,
        mark=x,
        ]
        coordinates {
        (1,6)(2,3)(3,2.3)(4,1.4)(5,0.75)(6,0.5)(7,0.25)(8,0.3)(9,0.3)(10,0.18)
        };
        \addlegendentry{$P_{\text{FP}_{\max}}=0.1$}
     \addplot[
        color=mycolor3,
        mark=triangle,
        ]
        coordinates {
        (1,5)(2,2.7)(3,1.5)(4,0.9)(5,0.5)(6,0.3)(7,0)(8,0)(9,0)(10,0)
        };
        \addlegendentry{$P_{\text{FP}_{\max}}=0.15$}

    \addplot[color=mycolor1,mark=square,dashed]
        coordinates {(1,10)(2,8)(3,7)(4,6)(5,5.5)(6,4)(7,3.5)(8,1.7)(9,0.5)(10,0.3)};
    \addplot[color=mycolor2,mark=x,dashed]
        coordinates {(1,8)(2,5)(3,4)(4,3.8)(5,3)(6,2.5)(7,2)(8,0.7)(9,0.5)(10,0.2)};
    \addplot[color=mycolor3,mark=triangle,dashed]
        coordinates {(1,6)(2,3)(3,2.5)(4,2.5)(5,2.5)(6,2.2)(7,1.8)(8,0.7)(9,0.5)(10,0.2)};
    \end{axis}
    \end{tikzpicture}
    \caption{Signals of opportunity}
    \end{subfigure}
    \begin{subfigure}[b]{0.24\textwidth}
    \begin{tikzpicture}[scale=.5,font=\Large]
    \begin{axis}[
        xlabel={Attack deviation [m]},
        ylabel={Detection time delay [s]},
        xmin=1, xmax=10,
        ymin=0, ymax=15,
        xtick={1,2,3,4,5,6,7,8,9,10},
        ytick={0,3,6,9,12,15},
        legend cell align={left},
        legend pos=north east,
        legend columns=1, 
        xmajorgrids=true,
        ymajorgrids=true,
        grid style=dashed,
    ]
    \addplot[color=mycolor6,mark=square]
        coordinates {(1,7.5)(2,5)(3,4.5)(4,3.3)(5,3)(6,2.1)(7,1.3)(8,0.5)(9,0.2)(10,0.1)};
        \addlegendentry{Proposed with Wi-Fi}
    \addplot[color=mycolor4,mark=o]
        coordinates {(1,8.5)(2,5.3)(3,5)(4,3.9)(5,3.3)(6,2.6)(7,2)(8,1)(9,0.6)(10,0.4)};
        \addlegendentry{Signals of opportunity}
    \end{axis}
    \end{tikzpicture}
    \caption{Only Wi-Fi, $P_{\text{FP}_{\max}}=0.1$}
    \end{subfigure}
    \hfill
    \begin{subfigure}[b]{0.24\textwidth}
    \begin{tikzpicture}[scale=.5,font=\Large]
    \begin{axis}[
        xlabel={Attack deviation [m]},
        ylabel={Detection time delay [s]},
        xmin=1, xmax=10,
        ymin=0, ymax=15,
        xtick={1,2,3,4,5,6,7,8,9,10},
        ytick={0,3,6,9,12,15},
        legend cell align={left},
        legend pos=north east,
        legend columns=1, 
        xmajorgrids=true,
        ymajorgrids=true,
        grid style=dashed,
    ]
    \addplot[color=mycolor6,mark=square]
        coordinates {(1,7.1)(2,4.9)(3,4.2)(4,3.1)(5,2.6)(6,2.1)(7,1.2)(8,0.4)(9,0.1)(10,0.1)};
        \addlegendentry{Proposed with cellular}
    \addplot[color=mycolor4,mark=o]
        coordinates {(1,8.3)(2,5.4)(3,4.9)(4,3.8)(5,3)(6,2.6)(7,2.2)(8,1.4)(9,0.6)(10,0.3)};
        \addlegendentry{Signals of opportunity}
    \end{axis}
    \end{tikzpicture}
    \caption{Only cellular, $P_{\text{FP}_{\max}}=0.1$}
    \end{subfigure}
    \caption{(Case 1) $\Delta T$ for PDS using network-based positioning and signals of opportunity method \cite{OliSciIbrDip:J22} as a function of the spoofing deviation. Solid lines are for augmented Datasets A and B, and dashed lines are for augmented Dataset C.}
    \label{fig:latdevnet}
\end{figure*}
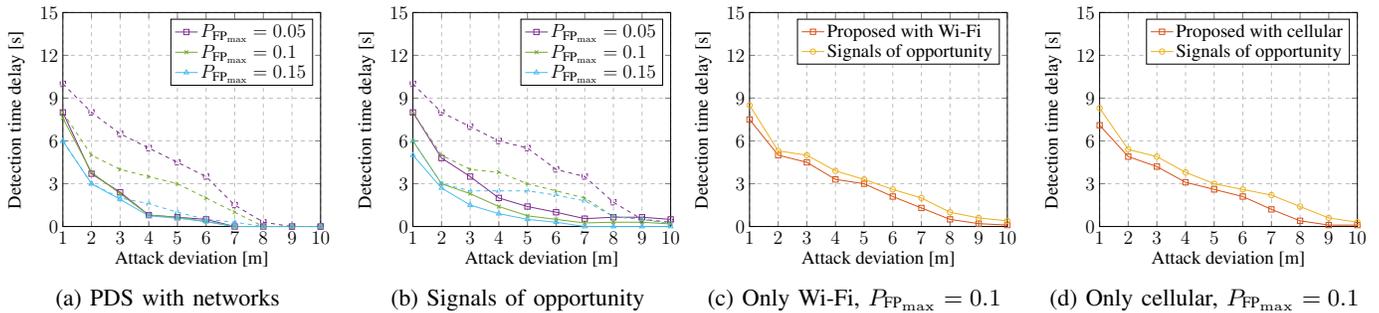

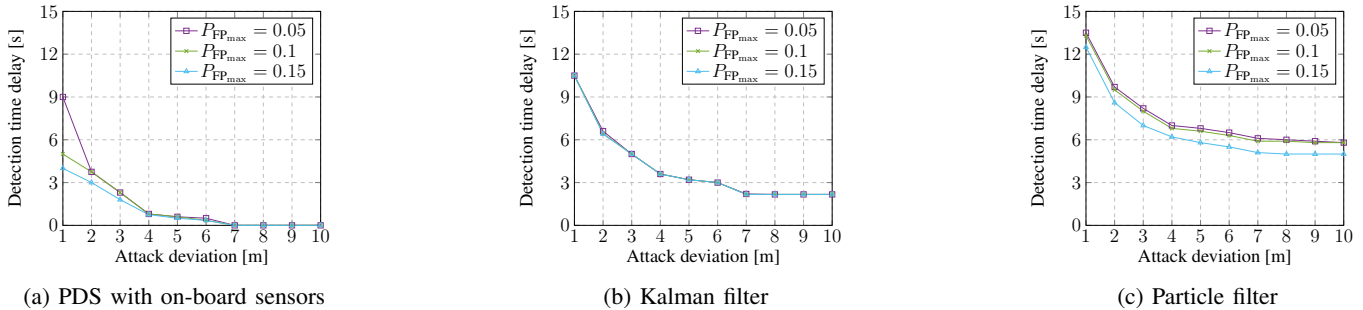
\begin{figure*}
    \centering
    \begin{subfigure}[b]{0.25\textwidth}
    \begin{tikzpicture}[scale=.5,font=\Large]
    \begin{axis}[
        xlabel={Attack deviation [m]},
        ylabel={Detection time delay [s]},
        xmin=1, xmax=10,
        ymin=0, ymax=15,
        xtick={1,2,3,4,5,6,7,8,9,10},
        ytick={0,3,6,9,12,15},
        legend cell align={left},
        legend pos=north east,
        legend columns=1, 
        xmajorgrids=true,
        ymajorgrids=true,
        grid style=dashed,
    ]
    \addplot[
        color=mycolor1,
        mark=square,
        ]
        coordinates {
        (1,9)(2,3.75)(3,2.3)(4,0.8)(5,0.6)(6,0.5)(7,0)(8,0)(9,0)(10,0)
        };
        \addlegendentry{$P_{\text{FP}_{\max}}=0.05$}
     \addplot[
        color=mycolor2,
        mark=x,
        ]
        coordinates {
        (1,5)(2,3.75)(3,2.3)(4,0.8)(5,0.6)(6,0.35)(7,0)(8,0)(9,0)(10,0)
        };
        \addlegendentry{$P_{\text{FP}_{\max}}=0.1$}
     \addplot[
        color=mycolor3,
        mark=triangle,
        ]
        coordinates {
        (1,4)(2,3)(3,1.8)(4,0.75)(5,0.5)(6,0.35)(7,0)(8,0)(9,0)(10,0)
        };
        \addlegendentry{$P_{\text{FP}_{\max}}=0.15$}
    \end{axis}
    \end{tikzpicture}
    \caption{PDS with on-board sensors}
    \end{subfigure}
    \hfill
    \begin{subfigure}[b]{0.25\textwidth}
    \begin{tikzpicture}[scale=.5,font=\Large]
    \begin{axis}[
        xlabel={Attack deviation [m]},
        ylabel={Detection time delay [s]},
        xmin=1, xmax=10,
        ymin=0, ymax=15,
        xtick={1,2,3,4,5,6,7,8,9,10},
        ytick={0,3,6,9,12,15},
        legend cell align={left},
        legend pos=north east,
        legend columns=1, 
        xmajorgrids=true,
        ymajorgrids=true,
        grid style=dashed,
    ]
    \addplot[
        color=mycolor1,
        mark=square,
        ]
        coordinates {
        (1,10.5)(2,6.6)(3,5)(4,3.6)(5,3.2)(6,3)(7,2.2)(8,2.17)(9,2.17)(10,2.17)
        };
        \addlegendentry{$P_{\text{FP}_{\max}}=0.05$}
     \addplot[
        color=mycolor2,
        mark=x,
        ]
        coordinates {
        (1,10.5)(2,6.43)(3,5)(4,3.6)(5,3.2)(6,3)(7,2.17)(8,2.17)(9,2.17)(10,2.17)
        };
        \addlegendentry{$P_{\text{FP}_{\max}}=0.1$}
     \addplot[
        color=mycolor3,
        mark=triangle,
        ]
        coordinates {
        (1,10.5)(2,6.4)(3,5)(4,3.6)(5,3.2)(6,3)(7,2.17)(8,2.17)(9,2.17)(10,2.17)
        };
        \addlegendentry{$P_{\text{FP}_{\max}}=0.15$}
    \end{axis}
    \end{tikzpicture}
    \caption{Kalman filter}
    \end{subfigure}
    \hfill
    \begin{subfigure}[b]{0.25\textwidth}
    \begin{tikzpicture}[scale=.5,font=\Large]
    \begin{axis}[
        xlabel={Attack deviation [m]},
        ylabel={Detection time delay [s]},
        xmin=1, xmax=10,
        ymin=0, ymax=15,
        xtick={1,2,3,4,5,6,7,8,9,10},
        ytick={0,3,6,9,12,15},
        legend cell align={left},
        legend pos=north east,
        legend columns=1, 
        xmajorgrids=true,
        ymajorgrids=true,
        grid style=dashed,
    ]
    \addplot[
        color=mycolor1,
        mark=square,
        ]
        coordinates {
        (1,13.5)(2,9.7)(3,8.2)(4,7)(5,6.8)(6,6.5)(7,6.1)(8,6)(9,5.9)(10,5.8)
        };
        \addlegendentry{$P_{\text{FP}_{\max}}=0.05$}
     \addplot[
        color=mycolor2,
        mark=x,
        ]
        coordinates {
        (1,13.3)(2,9.5)(3,8)(4,6.8)(5,6.6)(6,6.3)(7,5.9)(8,5.9)(9,5.8)(10,5.8)
        };
        \addlegendentry{$P_{\text{FP}_{\max}}=0.1$}
     \addplot[
        color=mycolor3,
        mark=triangle,
        ]
        coordinates {
        (1,12.5)(2,8.6)(3,7)(4,6.2)(5,5.8)(6,5.5)(7,5.1)(8,5)(9,5)(10,5)
        };
        \addlegendentry{$P_{\text{FP}_{\max}}=0.15$}
    \end{axis}
    \end{tikzpicture}
    \caption{Particle filter}
    \end{subfigure}
    \caption{(Case 2) $\Delta T$ for PDS, Kalman filter, and particle filter, as a function of the spoofing deviation for augmented Datasets A and B. Note that the lower bound is 0 because we exclude computation delay from the computation machine.}
    \label{fig:latdevmot}
\end{figure*}

\begin{figure}
    \centering
    \begin{subfigure}[b]{0.24\textwidth}
    \begin{tikzpicture}[scale=.5,font=\Large]
    \begin{axis}[
        xlabel={Attack deviation [m]},
        ylabel={Detection time delay [s]},
        xmin=1, xmax=10,
        ymin=0, ymax=15,
        xtick={1,2,3,4,5,6,7,8,9,10},
        ytick={0,3,6,9,12,15},
        legend cell align={left},
        legend pos=north east,
        legend columns=1, 
        xmajorgrids=true,
        ymajorgrids=true,
        grid style=dashed,
    ]
    \addplot[
        color=mycolor1,
        mark=square,
        ]
        coordinates {
        (1,8.5)(2,4.3)(3,2.8)(4,1)(5,0.6)(6,0.4)(7,0)(8,0)(9,0)(10,0)
        };
        \addlegendentry{$P_{\text{FP}_{\max}}=0.05$}
     \addplot[
        color=mycolor2,
        mark=x,
        ]
        coordinates {
        (1,8)(2,3.9)(3,2.2)(4,0.75)(5,0.5)(6,0.3)(7,0)(8,0)(9,0)(10,0)
        };
        \addlegendentry{$P_{\text{FP}_{\max}}=0.1$}
     \addplot[
        color=mycolor3,
        mark=triangle,
        ]
        coordinates {
        (1,6)(2,3)(3,1.8)(4,0.75)(5,0.5)(6,0.3)(7,0)(8,0)(9,0)(10,0)
        };
        \addlegendentry{$P_{\text{FP}_{\max}}=0.15$}
    \end{axis}
    \end{tikzpicture}
    \caption{PDS with all sources}
    \end{subfigure}
    \hfill
    \begin{subfigure}[b]{0.24\textwidth}
    \begin{tikzpicture}[scale=.5,font=\Large]
    \begin{axis}[
        xlabel={Attack deviation [m]},
        ylabel={Detection time delay [s]},
        xmin=1, xmax=10,
        ymin=0, ymax=15,
        xtick={1,2,3,4,5,6,7,8,9,10},
        ytick={0,3,6,9,12,15},
        legend cell align={left},
        legend pos=north east,
        legend columns=1, 
        xmajorgrids=true,
        ymajorgrids=true,
        grid style=dashed,
    ]
    \addplot[
        color=mycolor1,
        mark=square,
        ]
        coordinates {
        (1,10)(2,7)(3,4.4)(4,2.5)(5,2)(6,1.8)(7,1.5)(8,1.2)(9,1.2)(10,0.8)
        };
        \addlegendentry{$P_{\text{FP}_{\max}}=0.05$}
     \addplot[
        color=mycolor2,
        mark=x,
        ]
        coordinates {
        (1,8)(2,5.6)(3,3.8)(4,2.2)(5,1.8)(6,1.6)(7,1.25)(8,0.8)(9,0.7)(10,0.3)
        };
        \addlegendentry{$P_{\text{FP}_{\max}}=0.1$}
     \addplot[
        color=mycolor3,
        mark=triangle,
        ]
        coordinates {
        (1,6.3)(2,4)(3,3.3)(4,2)(5,1.7)(6,1.5)(7,1.2)(8,0.7)(9,0.7)(10,0.3)
        };
        \addlegendentry{$P_{\text{FP}_{\max}}=0.15$}
    \end{axis}
    \end{tikzpicture}
    \caption{Combined metrics}
    \end{subfigure}
    \caption{(Case 3) $\Delta T$ for PDS and the combined metrics algorithm, as a function of the spoofing deviation for augmented Datasets A and B.}
    \label{fig:latdevall}
\end{figure}
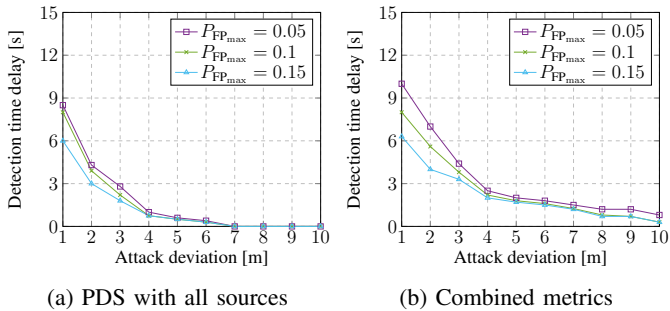
\revadd{Detection time delay is the time between the attack launch and detection.} Especially because the spoofing attacks in the \ac{gnss} traces are stealthy, gradually changing the induced deviation from the actual position, it is interesting to consider the time delay to detect the attack. We exclude computation delays and fix the false positive rate at $\{5,10,15\}$ percent. \revadd{Other experiment settings are the same as for the previous subsection.} We only measure time delay for successful attack detections.

Figs.~\ref{fig:latdevnet}--\ref{fig:latdevall} \revadd{show} the average $\Delta T$ as a function of deviation for cases 1--3 respectively. In Fig.~\ref{fig:latdevnet}, PDS has similar shapes of performance curves but always lower delay, because network locations are noisy and PDS can smoothen out the noise, thus deriving a better estimation of the actual location and being more sensitive to deviation. In Fig.~\ref{fig:latdevmot}, 
Kalman and particle filters have a delay that is always larger than 2--6 seconds, as they need time to update posterior distributions. Overall, our schemes are slightly better than the baseline ones for deviations from 1 to 5 meters. The detection delay of PDS is lower than 0.5 seconds when the deviation is larger than 5 meters. It also has at most 6 seconds performance gain compared to the other methods. Higher detection accuracy and more accurate alternative position result in lower $\Delta T$, because methods that provide better information for decision-making do help detect spoofing. In contrast, methods that do not fuse opportunistic information have a higher probability of missed detection, so they are not as fast as the proposed one. 

\subsection{Alternative Position Accuracy}
\begin{figure*}
\centering
\begin{subfigure}[b]{0.25\textwidth}
\begin{tikzpicture}[scale=.5,font=\Large]
    \begin{axis}[  
        xtick={1,2,3},
        xticklabels={Mean,Top 20\%,Bottom 20\%},
        xlabel={Metrics distribution},
    	ylabel={Absolute error [m]},
        ymin=0,ymax=125,xmin=0.5,xmax=3.5,
    	ybar=3pt,
        bar width=15pt,
        legend cell align={left},
        legend pos=north east,
        xmajorgrids=true,
        ymajorgrids=true,
        grid style=dashed,
    ]
    \addplot[fill=mycolor1]
    	coordinates {(1,12.8) (2,19.0) (3,5.3)};
    \addplot[fill=mycolor2]
    	coordinates {(1,35.3) (2,24.2) (3,12.4)};
    \legend{PDS,Signals}
    \end{axis}
    \end{tikzpicture}
    \caption{Case 1}
    \label{fig:maeposnet}
\end{subfigure}
\hfill
\begin{subfigure}[b]{0.25\textwidth}
\begin{tikzpicture}[scale=.5,font=\Large]
    \begin{axis}[
        xtick={1,2,3},
        xticklabels={Mean,Top 20\%,Bottom 20\%},
        xlabel={Metrics distribution},
    	ylabel={Absolute error [m]},
        ymin=0,ymax=125,xmin=0.5,xmax=3.5,
    	ybar=3pt,
        bar width=15pt,
        legend cell align={left},
        legend pos=north east,     
        xmajorgrids=true,
        ymajorgrids=true,
        grid style=dashed,
    ]
    \addplot[fill=mycolor1]
    	coordinates {(1,84.3) (2,36.1) (3,5.3)};
    \addplot[fill=mycolor5]
    	coordinates {(1,115) (2,72.1) (3,19.8)};
    \addplot[fill=mycolor3]
    	coordinates {(1,109) (2,69.2) (3,18.4)};
    \legend{PDS,PF,Kalman}
    \end{axis}
    \end{tikzpicture}
    \caption{Case 2}
    \label{fig:maeposmot}
\end{subfigure}
\hfill
\begin{subfigure}[b]{0.25\textwidth}
\begin{tikzpicture}[scale=.5,font=\Large]
    \begin{axis}[  
        xtick={1,2,3},
        xticklabels={Mean,Top 20\%,Bottom 20\%},
        xlabel={Metrics distribution},
    	ylabel={Absolute error [m]},
        ymin=0,ymax=125,xmin=0.5,xmax=3.5,
    	ybar=3pt,
        bar width=15pt,
        legend cell align={left},
        legend pos=north east,
        xmajorgrids=true,
        ymajorgrids=true,
        grid style=dashed,
    ]
    \addplot[fill=mycolor1]
    	coordinates {(1,12.7) (2,18.9) (3,5.1)};
    \addplot[fill=mycolor4]
    	coordinates {(1,19.75) (2,29.28) (3,8.7)};
    \legend{PDS,GLRT}
    \end{axis}
    \end{tikzpicture}
    \caption{Case 3}
    \label{fig:maeposall}
\end{subfigure}
\caption{Absolute error evaluation of alternative position accuracy over different schemes and traces.}
\end{figure*}
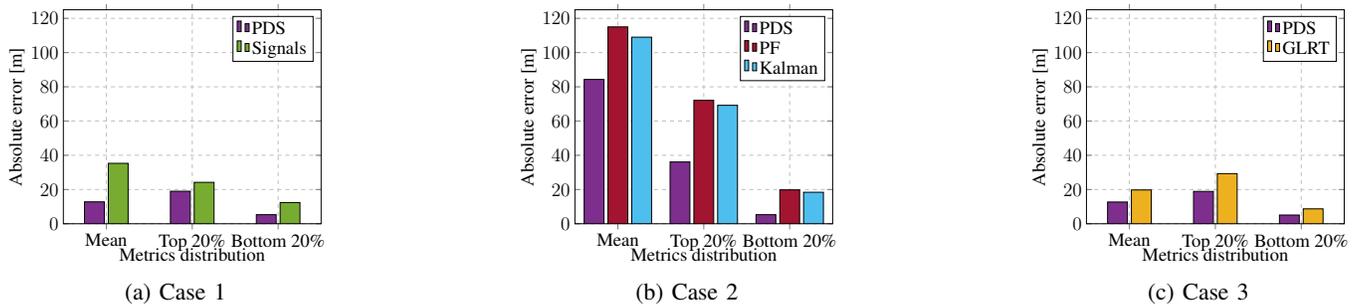
We define the alternative position as the combined mean of the confidence intervals in Eq.~\eqref{eq:poscom} independently of whether the \ac{gnss} receiver is under attack. In the dataset, the deviation of successful \ac{gnss} spoofing ranges from 0 to 1000 meters, and most of them are less than 50 meters. Here we calculate the absolute error of the true position versus the alternative position. The experimental parameters are identical to those used in the previous subsection.

For case 1, Fig.~\ref{fig:maeposnet} shows the \ac{mae} of alternative position without on-board sensors. We show the top 20\%, bottom 20\% and mean error of the distribution and plot them side by side. PDS has a much lower error for case 1, especially for the average error. For case 2, Fig.~\ref{fig:maeposmot} shows the \ac{mae} over different schemes and positions. For case 3, Fig.~\ref{fig:maeposall} shows the \ac{mae} when using all sources, including on-board sensors and network interfaces to calculate the alternative position. Overall, our alternative positions have smaller errors, which is only 14--77\% of the other methods \ac{mae}, due to combining different opportunistic information.

\subsection{Discussion}
\subsubsection{Performance}
PDS outperforms baseline methods in Sec.~\ref{numres}, because it considers contextual information and the correlation of locations. It fuses all sources and it naturally improves over the network-only or \ac{imu}-only variants. It reduces accumulated errors of on-board sensors and one-time errors of network-based positioning, so covers both long-term and short-term errors. Our PDS largely solves the gradual and strongest \ac{gnss} spoofing attacks in \cite{SheWonCheChe:C20}. 
\subsubsection{Uncertainty}
We compute mean and variance with Eq.~\eqref{eq:varcom} and \eqref{eq:poscom} to obtain confidence intervals. Then, we use mean and variance to compute the likelihood for decision-making. We conduct a comparison between the performance of PDS with uncertainty and its performance without uncertainty to determine the level of performance gain achieved. For the likelihood function, the scheme without uncertainty sets variance to a constant. We observe that PDS with uncertainty has \revadd{0.5--2\%} true positive gain over PDS without uncertainty. 
\subsubsection{Efficiency}
The computational complexity of PDS depends on the window size. In the local polynomial and Gaussian process regressions, matrix inversion costs most of the computations, and its complexity is $\mathcal{O}(w^3)$. There is a significant amount of work on accelerating matrix inversion. The choice of window size is a trade-off. If tiny, we gain efficiency at the expense of detection accuracy. If large, the algorithm will process needless historical data and be very slow. The window size could also be different in local polynomial and Gaussian process regressions. The trade-off needs to be explored in the future.

\section{Conclusion}
\label{conclu}
This paper develops a construction algorithm for confidence intervals and fuses opportunistic information to obtain the likelihood of \ac{gnss} being under attack. To decide if \ac{gnss} spoofing is underway, the core idea is to exploit the motion of the mobile platform and the statistical properties of the location data. First, we use a polynomial regression with motion constraints, proved to be convex, to estimate the location. Then, by using Gaussian process regression, we model the uncertainty of the location prediction, then fuse them into a likelihood function for probabilistic detection. The proposed detector has more than 7\% performance gain on average in true positive probability, lower detection time delay, and reduces at least 23\% error of alternative positions. 
In future work, we will assume network infrastructures that are also under attack and extend our scheme to explore more details in individual anchors.


\bibliographystyle{IEEEtran}
\bibliography{reference/references.bib}

\end{document}

%% file: acronym.tex
\newac{speb}{SPEB}{square position error bound}
\newac[plural=EFIMs,firstplural=Fisher information matrices (EFIMs)]{efim}{EFIM}{Fisher information matrix}
\newac{ne}{NE}{Nash equilibrium}
\newac{mse}{MSE}{mean squared error}
\newac{toa}{TOA}{time-of-arrival}
\newac{snr}{SNR}{signal-to-noise ratio}
\newac{lan}{LAN}{local area network}
\newac{psd}{PSD}{positive semidefinite}
\newac{pd}{PD}{positive definite}
\newac{wrt}{w.r.t.}{with respect to}
\newac{lhs}{L.H.S.}{left hand side}
\newac{wp1}{w.p.1}{with probability 1}
\newac{kkt}{KKT}{Karush-Kuhn-Tucker}
\newac{wlog}{w.l.o.g.}{without loss of generality}
\newac{mle}{MLE}{maximum likelihood estimation}
\newac{gps}{GPS}{Global Positioning System}
\newac{rssi}{RSSI}{received signal strength indicator}
\newac{mimo}{MIMO}{multiple-input multiple-output}
\newac{csi}{CSI}{channel state information}
\newac{fdd}{FDD}{frequency division duplexing}
\newac{ms}{MS}{mobile station}
\newac{bs}{BS}{base station}
\newac{d2d}{D2D}{device-to-device}
\newac{slnr}{SLNR}{signal-to-interference-leakage-and-noise-ratio}
\newac{ula}{ULA}{uniform linear antenna array}
\newac{pas}{PAS}{power angular spectrum}
\newac{mmse}{MMSE}{minimum mean square error}
\newac{zf}{ZF}{zero-forcing}
\newac{rzf}{RZF}{regularized zero-forcing}
\newac{as}{AS}{angular spread}
\newac{aod}{AOD}{angle of departure}
\newac{iid}{i.i.d.}{independent and identically distributed} 
\newac{sinr}{SINR}{signal-to-interference-and-noise ratio}
\newac{tdd}{TDD}{time-division duplex}
\newac{rvq}{RVQ}{random vector quantization}
\newac{rhs}{R.H.S.}{right hand side}
\newac{mrc}{MRC}{maximum ratio combining}
\newac{cdf}{CDF}{cumulative distribution function}
\newac{a.s.}{a.s.}{almost surely}
\newac{los}{LOS}{line-of-sight}
\newac{jsdm}{JSDM}{joint spatial division and multiplexing}
\newac{map}{MAP}{maximum a posteriori}
\newac{klt}{KLT}{Karhunen-Lo\`eve Transform}
\newac{lbe}{LBE}{link bargaining equilibrium}
\newac{se}{SE}{Stackelberg equilibrium}
\newac{uav}{UAV}{unmanned aerial vehicle}
\newac{nlos}{NLOS}{non-line-of-sight}
\newac{pdf}{PDF}{probability density function}
\newac{em}{EM}{expectation-maximization}
\newac{knn}{KNN}{$k$-nearest neighbor}
\newac{svd}{SVD}{singular value decomposition}
\newac{nmf}{NMF}{non-negative matrix factorization}
\newac{umf}{UMF}{unimodality-constrained matrix factorization}
\newac{rmse}{RMSE}{rooted mean squared error}
\newac{olos}{OLOS}{obstructed line-of-sight}
\newac{mmw}{mmW}{millimeter wave}
\newac{ber}{BER}{bit error rate}
\newac{rss}{RSS}{received signal strength}
\newac{lp}{LP}{linear program}
\newac{ufw}{U-FW}{unimodal Frank-Wolfe}
\newac{utf}{UTF}{unimodality-constrained tensor factorization}
\newac{fw}{FW}{Frank-Wolfe}
\newac{iot}{IoT}{Internet-of-Things}
\newac{mae}{MAE}{mean absolute error}
\newac{crb}{CRB}{Cram\'er-Rao bound}
\newac{aoa}{AoA}{angle of arrival}
\newac{wcl}{WCL}{weighted centroid localization}
\newac[plural=GNSS,firstplural=Global Navigation Satellite Systems (GNSS)]{gnss}{GNSS}{Global Navigation Satellite System}
\newac{gsm}{GSM}{global system for mobile communications}
\newac{imu}{IMU}{inertial measurement unit}
\newac{rbf}{RBF}{radial basis function}
\newac{msf}{MSF}{multi-sensor fusion}
\newac{lidar}{LiDAR}{light detection and ranging}
\newac{glrt}{GLRT}{generalized likelihood ratio test}
\newac{sdr}{SDR}{software-defined radio}
\newac{ap}{AP}{access point}
\newac{lte}{LTE}{Long-Term Evolution}
\newac{raim}{RAIM}{Receiver Autonomous Integrity Monitoring}
\newac{pvt}{PVT}{position, velocity and time}
\newac{npl}{NPL}{Neyman-Pearson lemma}
\newac{sop}{SOP}{signals of opportunity}